\documentclass[11pt, fleqn]{article}
\usepackage[english]{babel}

\usepackage[lmargin=1.1in,rmargin=1.1in,bottom=1.3in,top=1.3in,
twoside=False]{geometry}

\usepackage{relsize,xspace}
 \usepackage{xcolor}
 \usepackage{mathtools}
 \usepackage{todonotes}
 \usepackage{comment}
\usepackage{microtype}
\usepackage{amsmath}
\usepackage{amssymb}
\usepackage{amsfonts}
\usepackage{stmaryrd}
\usepackage{tikz}
\usepackage{bm}
\usepackage{tikz}
\usepackage{refcount}
\usepackage{wrapfig}

\usepackage{marginnote}

\definecolor{blue}{rgb}{0.1,0.2,0.5}
\definecolor{brown}{rgb}{0.6,0.6,0.2}
\usepackage[ocgcolorlinks, linkcolor={blue}, citecolor={brown}]{hyperref}

\usepackage[amsmath,thmmarks,hyperref]{ntheorem}
\usepackage{cleveref}

\crefformat{page}{#2page~#1#3}%
\Crefformat{page}{#2Page~#1#3}%
\crefformat{equation}{#2(#1)#3}%
\Crefformat{equation}{#2(#1)#3}%
\crefformat{figure}{#2Figure~#1#3}%
\Crefformat{figure}{#2Figure~#1#3}%
\crefformat{section}{#2Section~#1#3}
\Crefformat{section}{#2Section~#1#3}
\crefformat{chapter}{#2Chapter~#1#3}
\Crefformat{chapter}{#2Chapter~#1#3}
\crefformat{chapter*}{#2Chapter~#1#3}
\Crefformat{chapter*}{#2Chapter~#1#3}
\crefformat{part}{#2Part~#1#3}
\Crefformat{part}{#2Part~#1#3}
\crefformat{enumi}{#2(#1)#3}
\Crefformat{enumi}{#2(#1)#3}

\usepackage{enumerate}

\usepackage{latexsym}

\theoremnumbering{arabic}
\theoremstyle{plain}
\theoremsymbol{}
\theorembodyfont{\itshape}
\theoremheaderfont{\normalfont\bfseries}
\theoremseparator{}

\newtheorem{theorem}{Theorem}
\crefformat{theorem}{#2Theorem~#1#3}
\Crefformat{theorem}{#2Theorem~#1#3}

\newcommand{\newtheoremwithcrefformat}[2]{%
  \newtheorem{#1}[lemma]{#2}%
  \crefformat{#1}{##2\MakeUppercase#1~##1##3}%
  \Crefformat{#1}{##2\MakeUppercase#1~##1##3}%
}
\newcommand{\newseptheoremwithcrefformat}[2]{%
  \newtheorem{#1}{#2}%
  \crefformat{#1}{##2\MakeUppercase#1~##1##3}%
  \Crefformat{#1}{##2\MakeUppercase#1~##1##3}%
}

\newseptheoremwithcrefformat{lemma}{Lemma}
\newtheoremwithcrefformat{proposition}{Proposition}
\newtheoremwithcrefformat{observation}{Observation}
\newtheoremwithcrefformat{conjecture}{Conjecture}
\newtheoremwithcrefformat{corollary}{Corollary}
\newseptheoremwithcrefformat{claim}{Claim}
\theorembodyfont{\upshape}
\newtheoremwithcrefformat{example}{Example}
\newtheoremwithcrefformat{remark}{Remark}
\newseptheoremwithcrefformat{definition}{Definition}

\theoremstyle{nonumberplain}
\theoremheaderfont{\scshape}
\theorembodyfont{\normalfont}
\theoremsymbol{\ensuremath{\square}}
\newtheorem{proof}{Proof.}

\theoremsymbol{\ensuremath{\lrcorner}}
\newtheorem{clproof}{Proof.}

\newcommand{\wcol}{\mathrm{wcol}}

\newcommand{\WReach}{\mathrm{WReach}}

\newcommand{\Oof}{\mathcal{O}}
\newcommand{\CCC}{\mathcal{C}}

\newcommand{\WWW}{\mathcal{W}}

\newcommand{\PPP}{\mathcal{P}}
\newcommand{\FFF}{\mathcal{F}}

\newcommand{\YYY}{\mathcal{Y}}
\newcommand{\nei}{\mathrm{nei}}
\renewcommand{\ker}{\mathrm{ker}}

\newcommand{\grad}{\nabla}
\newcommand{\ds}{\mathbf{ds}}
\newcommand{\cl}{\mathrm{cl}}
\newcommand{\cst}{\alpha}

\newcommand{\fnei}{f_{\nei}}
\newcommand{\fwcol}{f_{\wcol}}
\newcommand{\fker}{f_{\ker}}
\newcommand{\fproj}{f_{\mathrm{proj}}}
\newcommand{\fcl}{f_{\cl}}
\newcommand{\fgrad}{f_{\grad}}
\newcommand{\fpaths}{f_{\mathrm{pth}}}
\newcommand{\fapx}{f_{\mathrm{apx}}}
\newcommand{\fcore}{f_{\mathrm{core}}}
\newcommand{\ffin}{f_{\mathrm{fin}}}

\newcommand{\suchthat}{ \colon }

\newcommand{\N}{\mathbb{N}}

\renewcommand{\phi}{\varphi}
\renewcommand{\epsilon}{\varepsilon}

\newcommand{\minor}{\preccurlyeq}
\newcommand{\dist}{\mathrm{dist}}
\newcommand{\indx}{\mathrm{index}}
\renewcommand{\mid}{~:~}

\newcommand{\profnum}{\widehat{\nu}}
\newcommand{\projnum}{\mu}
\newcommand{\projprof}{\widehat{\mu}}

\newcommand{\abs}[1]{\ensuremath{\left\lvert#1\right\rvert}}

\title{Neighborhood complexity and kernelization\\ for nowhere dense classes of graphs\thanks{
The authors from Technische Universit\"at Berlin (AG, SK, OK, and RR)
have been supported by the
European Research Council (ERC) under the European Union's Horizon
2020 research and innovation programme (ERC consolidator grant DISTRUCT,
agreement No.\ 648527).
The work of the authors from University of Warsaw (MP and Seb.\ S) is supported by the National Science Centre of Poland via POLONEZ grant agreement UMO-2015/19/P/ST6/03998, 
which has received funding from the European Union's Horizon 2020 research and 
innovation programme (Marie Sk\l odowska-Curie grant agreement No.\ 665778).
M. Pilipczuk is supported by Foundation for Polish Science (FNP) via the START stipend programme.
 }}

\author{
     Kord Eickmeyer\thanks{Technische Universit\"at Darmstadt, Germany, \texttt{eickmeyer@mathematik.tu-darmstadt.de}}
\and Archontia C.\@ Giannopoulou\thanks{Technische Universit\"at Berlin, Germany, \texttt{\{firstname.surname : Authors\}@tu-berlin.de}}
\and Stephan Kreutzer$^\ddagger$
\and O-joung Kwon$^\ddagger$
\and Micha\l~Pilipczuk\thanks{Institute of Informatics, University of Warsaw, Poland, \texttt{\{michal.pilipczuk,siebertz\}@mimuw.edu.pl}}
\and Roman Rabinovich$^\ddagger$
\and Sebastian Siebertz$^\S$}

\begin{document}

\maketitle
\begin{abstract}
We prove that whenever $G$ is a graph from a nowhere dense graph class $\CCC$,
and $A$ is a subset of vertices of $G$, then the number of subsets of~$A$
that are realized as intersections of~$A$ with $r$-neighborhoods of vertices of $G$ is
at most $f(r,\epsilon)\cdot |A|^{1+\epsilon}$, where
$r$ is any positive integer, $\epsilon$ is any positive real, and $f$ is a function that depends only on the class $\CCC$.
This yields a characterization of nowhere dense classes of graphs in terms of {\em{neighborhood complexity}}, 
which answers a question posed by Reidl et al.~\cite{reidl2016characterising}.
As an algorithmic application of the above result, we show that for every fixed $r$, the parameterized 
{\sc{Distance-$r$ Dominating Set}} problem admits an almost linear kernel on any nowhere dense graph class.
This proves a conjecture posed by Drange et al.~\cite{drange2016kernelization}, and shows that the limit of
parameterized tractability of {\sc{Distance-$r$ Dominating Set}} on subgraph-closed graph classes lies exactly
on the boundary between nowhere denseness and somewhere denseness.
\end{abstract}

\begin{picture}(0,0) \put(395,-237)
{\hbox{\includegraphics[scale=0.25]{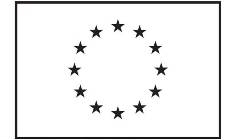}}} \end{picture} 
\vspace{-0.8cm}

\section{Introduction}

\paragraph*{Sparse graphs.}The notion of nowhere denseness was 
introduced by Ne\v set\v ril and
Ossona de Mendez~\cite{nevsetvril2011nowhere,nevsetvril2010first} as
a general model of \emph{uniform sparseness} of graphs. Many
familiar classes of sparse graphs, like planar
graphs, graphs of bounded treewidth, graphs of bounded degree, and,
in fact, all classes that exclude a fixed (topological) minor, are nowhere
dense. Notably, classes of bounded average degree or bounded
degeneracy are not necessarily nowhere dense. In an algorithmic context this is
reasonable, as every graph can be turned into a graph of
degeneracy at most~$2$ by subdividing every edge once; however, the
structure of the graph is essentially preserved under this operation.

Formally, a graph $H$ is a
\emph{minor} of a graph $G$, written $H\minor G$, if there are
pairwise vertex-disjoint connected subgraphs $(I_u)_{u\in V(H)}$ of~$G$
such that whenever $\{u,v\}$ is an edge in $H$, there are $u'\in I_u$ and $v'\in I_v$ for which $\{u',v'\}$ is an edge in $G$.
Such a family $(I_u)_{u\in V(H)}$ is called a {\em{minor model}} of $H$ in $G$.
The graph $H$ is a {\em{depth-$r$ minor}} of $G$, denoted $H\minor_rG$, if there is a minor model
$(I_u)_{u\in V(H)}$ of~$H$ in $G$ such that each subgraph $I_u$ has radius at most $r$.

\begin{definition}\label{def:nwd}
  A class $\CCC$ of graphs is \emph{nowhere dense} if there is a
  function $f\,\colon\,\N\rightarrow \N$ such that
  $K_{f(r)}\not\minor_r G$ for all $r\in \N$ and all $G\in \CCC$.
\end{definition}

Nowhere denseness turns out to be a very robust concept with several
seemingly unrelated natural characterizations.  These include
characterizations by the density of shallow (topological)
minors~\cite{nevsetvril2011nowhere,nevsetvril2010first},
quasi-wideness~\cite{nevsetvril2011nowhere} (a notion introduced by
Dawar~\cite{dawar2010homomorphism} in his study of homomorphism
preservation properties), low tree-depth
colorings~\cite{nevsetvril2008grad}, generalized coloring
numbers~\cite{zhu2009coloring}, sparse neighborhood
covers~\cite{GroheKRSS15,grohe2014deciding}, by a game called the
splitter game~\cite{grohe2014deciding} and by the model-theoretic
concepts of stability and independence~\cite{adler2014interpreting}.
For a broader discussion we refer to the book
of Ne\v{s}et\v{r}il and Ossona de Mendez~\cite{sparsity}.

This wealth of concepts and results has been applied to multiple areas in mathematics and computer science, 
but a particularly frutiful line of research concerns the design of efficient
algorithms for hard computational problems on specific sparse classes of
graphs that occur naturally in applications. For instance, using low
tree-depth colorings, Ne\v set\v ril and Ossona de
Mendez~\cite{nevsetvril2010first} showed that the {\sc{Subgraph Isomorphism}}
and {\sc{Homomorphism}} problems are fixed-parameter tractable on any nowhere dense
class, parameterized by the size of the pattern graph. 
It was later shown by Grohe et al.~\cite{grohe2014deciding}
that, in fact, every first-order definable problem can be decided in
almost linear time on any nowhere dense graph class.

An important and related concept is the notion of a graph class of \emph{bounded expansion}.
Precisely, a class of graphs $\CCC$ has bounded expansion if for any $r\in \N$, the
ratio between the numbers of edges and vertices in any $r$-shallow minor of a graph from $\CCC$ is bounded
by a constant depending on $r$ only.
Obviously, every class of bounded expansion is also nowhere dense, but the converse is not always true.
Classes of bounded expansion were also introduced by Ne\v set\v ril
and Ossona de Mendez~\cite{nevsetvril2008grad,nevsetvril2008gradb,nevsetvril2008gradc}, 
and in this setting even more algorithmic
applications are known. These include database query
answering and enumeration \cite{KazanaS13}, or approximation and kernelization algorithms for domination problems~\cite{drange2016kernelization,dvovrak2013constant}, 
a topic that we explore in detail next.

\paragraph*{Domination problems.}
In the parameterized {\sc{Dominating Set}} problem we are given a
graph~$G$ and an integer parameter $k$, and the task is
to determine the existence of a subset $D\subseteq V(G)$ of size at
most $k$ such that every vertex $u$ of $G$ is \emph{dominated} by
$D$, that is, $u$ either belongs to $D$ or has a neighbor in $D$.
More generally, for fixed $r\in \N$ we can consider the \textsc{Distance-$r$ Dominating Set} problem,
where we are asked to determine the existence of a subset $D\subseteq V(G)$ of size at most
$k$ such that every vertex $u\in V(G)$ is within distance at most~$r$
from a vertex from $D$. 
The \textsc{Dominating Set} problem plays a central role in the theory of
parameterized complexity, as it is a prime example of a
$\mathsf{W}[2]$-complete problem, thus considered intractable in full generality from the parameterized point of view.
For this reason, \textsc{Dominating Set} and \textsc{Distance-$r$ Dominating Set}
have been extensively studied in restricted graph classes, including the sparse setting.

The study of parameterized algorithms for {\sc{Dominating Set}} on sparse and topologically-constrained graph classes has a long history, 
and, arguably, it played a pivotal role in the development of modern parameterized complexity.
A point of view that was particularly fruitful, and most relevant to our work, is {\em{kernelization}}.
Recall that a kernelization algorithm is a polynomial-time preprocessing algorithm that transforms a given instance into an equivalent one whose size is bounded by a function
of the parameter only, independently of the overall input size.
We are mostly interested in kernelization algorithms whose output guarantees are polynomial in the parameter, or maybe even linear.
For {\sc{Dominating Set}} on topologically-restricted graph classes, linear kernels 
were given for planar graphs~\cite{alber2004polynomial}, bounded genus
graphs~\cite{bodfomlok+09}, apex-minor-free graphs~\cite{fomin10},
graphs excluding a fixed minor~\cite{fomin2012linear}, and
graphs excluding a fixed topological minor~\cite{FominLST13}.
All these results relied on applying tools of topological nature, most importantly deep decomposition theorems for graphs excluding a fixed (topological) minor.
Notably, the research on kernelization for {\sc{Dominating Set}} directly led to the introduction of the technique of {\em{meta-kernelization}}~\cite{bodfomlok+09}, which applies to a much larger family of problems
on bounded-genus and $H$-minor-free graph classes.

The first application of the abstract sparsity theory to domination problems, beyond the topologically-constrained setting, was given by 
Dawar and Kreutzer~\cite{DawarK09}, who showed that for every $r\in \N$ and every nowhere dense class $\CCC$,
\textsc{Distance-$r$ Dominating Set} is fixed-parameter
tractable on $\CCC$. This result is based on the characterization of
nowhere denseness in terms of uniform quasi-wideness. 
It is known that nowhere dense classes are the limit for the fixed-parameter tractability of the problem:
Drange et al.~\cite{drange2016kernelization} showed that whenever $\CCC$ is a somewhere dense class closed under taking subgraphs, there is some $r\in \N$ for which
\textsc{Distance-$r$ Dominating Set} is $\mathsf{W}[2]$-hard on $\CCC$.

As far as polynomial kernelization is concerned, 
Drange et al.~\cite{drange2016kernelization} gave a linear kernel for \textsc{Distance-$r$ Dominating Set} on any graph class of bounded expansion\footnote{Precisely, the kernelization
algorithm of Drange et al.~\cite{drange2016kernelization} outputs an instance of an annotated problem where some vertices are not required to be dominated; this will be the case in this paper as well.}, 
for every $r\in \N$,
and an almost linear kernel for \textsc{Dominating Set} on any nowhere dense graph class; that is, a kernel of size $f(\epsilon)\cdot k^{1+\epsilon}$ for some function $f$.
Drange et al. could not extend their techniques to larger domination radii $r$ on nowhere dense classes, however, they conjectured that this should be possible.
Together with the aforementioned negative result, this would confirm the following dichotomy conjecture for the parameterized complexity of \textsc{Distance-$r$ Dominating Set}.

\begin{conjecture}[\cite{drange2016kernelization}]\label{con:kernel}
Let $\CCC$ be a class of graphs closed under taking subgraphs.
If $\CCC$ is nowhere dense, then for each $r\in \N$, the \textsc{Distance-$r$ Dominating Set} problem admits an almost linear kernel on~$\CCC$.
Otherwise, if $\CCC$ is somewhere dense, then for some $r\in \N$, the \textsc{Distance-$r$ Dominating Set} problem is $\mathsf{W}[2]$-hard.
\end{conjecture}

An important step towards proving~\cref{con:kernel} was made recently by a subset of the authors~\cite{siebertz2016polynomial}, who gave
a polynomial kernel for \textsc{Distance-$r$ Dominating Set} on any nowhere dense class $\CCC$. 
This result, in fact, follows from the new polynomial bounds on uniform quasi-wideness, proved in~\cite{siebertz2016polynomial}, in combination with the old approach of Dawar and Kreutzer~\cite{DawarK09}.
However, the degree of the polynomial bound on the kernel size in~\cite{siebertz2016polynomial} depends badly on $r$ and the class we are working with.

\paragraph*{Neighborhood complexity.} 
One of the crucial ideas in the work of Drange et al.~\cite{drange2016kernelization}
was to focus on the {\em{neighborhood complexity}} in sparse graph classes.
More precisely, for an integer $r\in \N$, a graph $G$, and a subset
$A\subseteq V(G)$ of vertices of~$G$,
the {\em{$r$-neighborhood complexity}} of~$A$, denoted $\nu_r(G, A)$, is defined as the number of different
subsets of $A$ that are of the form $N_r[u]\cap A$ for some vertex $u$ of $G$; here, $N^G_r[u]$ denotes the ball of radius $r$ around $u$. 
That is,
\[\nu_r(G, A)=|\{N_r^G[u]\cap A\ \colon\ u\in V(G)\}|.\]
It was proved by Reidl et al.~\cite{reidl2016characterising} that linear neighborhood complexity
exactly characterizes subgraph-closed classes of bounded expansion.
More precisely, a subgraph-closed class $\CCC$ has bounded expansion if and only if for each $r\in \N$ there is a constant $c_r$ such
that $\nu_r(G, A)\leq c_r\cdot |A|$ for all graphs $G\in \CCC$ and vertex subsets $A\subseteq V(G)$.
They posed as an open problem whether nowhere denseness can be similarly characterized by almost linear neighborhood complexity, thus essentially stating the following conjecture.

\begin{conjecture}[\cite{reidl2016characterising}]\label{con:nei}
Let $\CCC$ be a graph class closed under taking subgraphs.
Then $\CCC$ is nowhere dense if and only if there exists a function $\fnei(r,\epsilon)$ such that
$\nu_r(G,A)\leq \fnei(r,\epsilon)\cdot |A|^{1+\epsilon}$ for all $r\in \N$, $\epsilon>0$, graphs $G\in \CCC$, and vertex subsets $A\subseteq V(G)$.
\end{conjecture}

It is easy to see the right-to-left implication of~\cref{con:nei}. If $\CCC$ is somewhere dense and closed under taking subgraphs, then for some $r\in \N$ it contains the 
exact $r$-subdivision
of every graph~\cite{nevsetvril2011nowhere}. In this case it is easy to see that the $r$-neighborhood complexity of a vertex subset~$A$ can be as large as $2^{|A|}$.
The lack of the left-to-right direction was a major, however not the only, obstacle preventing Drange et al.~\cite{drange2016kernelization} from extending their kernelization results to {\sc{Distance-$r$ Dominating Set}}
on any nowhere dense class. The proof of this direction for $r=1$ was given by Gajarsk\'y et al.~\cite{GajarskyHOORRVS13}, and this result was used by Drange et al.~\cite{drange2016kernelization}
in their kernelization algorithm for {\sc{Dominating Set}} (with distance $r=1$) on nowhere dense graph classes.

\paragraph*{Our results.}
In this paper we resolve in affirmative both~\cref{con:kernel} and~\cref{con:nei}. Let us start with the latter.

\begin{theorem}\label{thm:main-neighborhood-complexity}
Let $\CCC$ be a graph class closed under taking subgraphs.
Then $\CCC$ is nowhere dense if and only if there exists a function $\fnei(r,\epsilon)$ such that
$\nu_r(G,A)\leq \fnei(r,\epsilon)\cdot |A|^{1+\epsilon}$ for all $r\in \N$, $\epsilon>0$, $G\in \CCC$, and $A\subseteq V(G)$.
\end{theorem}

To prove the above, we carefully analyze the argument of Reidl et al.~\cite{reidl2016characterising} for linear neighborhood complexity in classes of bounded expansion.
This argument is based on the analysis of vertex orderings certifying the constant upper bound on the {\em{weak coloring number}} of any graph from a fixed class of bounded expansion.
In the nowhere dense setting, we only have an $n^{\epsilon}$ upper bound on the weak coloring number, and therefore the reasoning breaks whenever one tries to use a bound that is exponential
in this number. We circumvent this issue by applying tools based on model-theoretic properties of nowhere
dense classes of graphs. More precisely, we use the fact that every nowhere dense class is {\em{stable}} in the sense of Shelah, and hence every graph~$H$ which is obtained from a graph $G$ from a nowhere dense class via a first-order interpretation has bounded VC-dimension~\cite{adler2014interpreting}.
Then exponential blow-ups can be reduced to polynomial using the Sauer-Shelah lemma.
These tools were recently used by a subset of the authors to give polynomial bounds for uniform quasi-wideness~\cite{siebertz2016polynomial}, 
a fact that also turns out to be useful in our proof.

We remark that we were informed by Micek, Ossona de Mendez, Oum, and Wood~\cite{Mendez17} that they 
have independently resolved~\cref{con:nei} by proving the same statement of~\cref{thm:main-neighborhood-complexity}, however using different methods.

Having the almost linear neighborhood complexity for any nowhere dense class of graphs, we can revisit the argumentation of Drange et al.~\cite{drange2016kernelization}
and complete the proof of~\cref{con:kernel} by showing the following result.

\begin{theorem}\label{thm:main-kernel}
  Let $\CCC$ be a fixed nowhere dense class of graphs, let $r$ be a fixed positive integer, and let $\epsilon>0$ be any fixed real.
  Then there is a polynomial-time algorithm that, given a graph $G\in \CCC$ and a positive integer $k$,
  returns a subgraph $G'\subseteq G$ and a vertex subset $Z\subseteq V(G')$ with the following properties:
  \begin{itemize}
  \item there is a set $D\subseteq V(G)$ of size at most $k$ which 
  $r$-dominates $G$ if and only if there is a set $D'\subseteq V(G')$
  of size at most~$k$ which $r$-dominates $Z$ in $G'$; and
  \item $|V(G')|\leq \fker(r,\epsilon)\cdot k^{1+\epsilon}$, for some function $\fker(r,\epsilon)$ depending only on the class $\CCC$.
  \end{itemize}
\end{theorem}

Just as in Drange et al.~\cite{drange2016kernelization}, the obtained triple $(G',Z,k)$ is formally not an instance of {\sc{Distance-$r$ Dominating Set}},
but of an annotated variant of this problem where some vertices (precisely $V(G')\setminus Z$) are not required to be dominated.
This is an annoying formal detail, however it can be addressed as follows.
As observed by Drange et al., such annotated instance can be reduced back to the standard problem in the following way:
add two fresh vertices $w,w'$, connect them by a path of length $r$, and connect $w$ with each vertex of $V(G')\setminus Z$ using a path of length $r$.
Then the obtained graph $G''$ has a dominating set of size at most $k+1$ if and only if $G'$ admits a set of at most $k$ vertices that $r$-dominates $Z$.
Drange et al.\ argued that this modification increases edge densities of $d$-shallow minors of the graph by at most $1$, for each $d\in \N$.
It is straightforward to see that the same modification in the nowhere dense setting may increase the size of the largest clique found as a depth-$d$ minor by at most $1$, for each $d\in \N$.
Thus, the graph $G''$ obtained in the reduction still belongs to a nowhere dense class $\CCC'$, where the sizes of excluded shallow clique minors, at all depths, are one larger than for $\CCC$.

Our proof of~\cref{thm:main-kernel} revisits the line of reasoning of Drange et al.~\cite{drange2016kernelization} for bounded expansion classes, and improves it in several places where the arguments
could not be immediately lifted to the nowhere dense setting.
The key ingredient is, of course, the usage of the newly proven almost linear neighborhood complexity for nowhere dense classes (\cref{thm:main-neighborhood-complexity}), however, this was not the only piece missing.
Another issue was that the algorithm of Drange et al.\@ starts by iteratively applying the constant-factor approximation algorithm for {\sc{Distance-$r$ Dominating Set}} of Dvo\v{r}\'ak~\cite{dvovrak2013constant} in order
to get very strong structural properties of the instance; in essence, they exploit also a large $2r$-independent set that Dvo\v{r}\'ak's algorithm provides in case of failure.
For the argument to work, it is crucial that this iteration finishes after a constant number of rounds, as each round introduces a constant multiplicative blow-up of the kernel size. However,
the argument for this is crucially based on the class having bounded expansion, and seems hard to lift to the nowhere dense setting.
We circumvent this problem by removing the whole iterative scheme of Drange et al.~\cite{drange2016kernelization} whatsoever, and replacing it with a simple $\Oof(\log \mathsf{OPT})$-approximation
following from the fact that the class we are working with has bounded VC-dimension. As a result, we do not obtain the large $2r$-independent set needed by Drange 
et al.\ for further stages, 
but we are able to extract a sufficient
substitute for it using the new polynomial bounds for uniform quasi-wideness~\cite{siebertz2016polynomial}. Thus, the whole algorithm is actually conceptually simpler than that of Drange et al., but this
comes at the cost of using deeper black-boxes relying on stability theory.

\paragraph*{Organization.} In~\cref{sec:prelim} we give the
necessary background about nowhere dense classes of graphs, and we introduce key definitions. 
\cref{sec:neighborhood-complexity} is devoted to the proof
of~\cref{thm:main-neighborhood-complexity}. In~\cref{sec:kernel} we
prove~\cref{thm:main-kernel}. \cref{sec:conclusion} contains concluding remarks.

\section{Preliminaries}\label{sec:prelim}

\paragraph*{Notation.}
By $\N$ we denote the set of non-negative integers.
For a set $A$, by $2^A$ we denote the family of all subsets of $A$.
For an equivalence relation $\equiv$, by $\indx(\equiv)$ we denote the number of equivalence classes of $\equiv$.

We use standard graph notation; see
e.g.~\cite{diestel2012graph} for reference. 
All graphs considered in
this paper are finite, simple, and undirected.  For a graph $G$, by
$V(G)$ and $E(G)$ we denote the vertex and edge sets of $G$,
respectively.  A graph~$H$ is a \emph{subgraph} of~$G$, denoted
$H\subseteq G$, if $V(H)\subseteq V(G)$ and $E(H)\subseteq E(G)$. 
A class of graphs $\CCC$ is \emph{monotone} if for every $G\in \CCC$ and every subgraph $H\subseteq G$, we also have $H\in \CCC$.
For
any $M\subseteq V(G)$, by $G[M]$ we denote the subgraph induced 
by~$M$.  We write $G-M$ for the graph $G[V(G)\setminus M]$ and if
$M=\{v\}$, we write $G-v$ instead. 

For a non-negative integer
$\ell$, a \emph{path of length $\ell$} in~$G$ is a sequence
$P=(v_1,\ldots, v_{\ell+1})$ of pairwise distinct vertices such that
$\{v_i,v_{i+1}\}\in E(G)$ for all $1\leq i\leq \ell$. We write $V(P)$ for
the vertex set $\{v_1,\ldots, v_{\ell+1}\}$ of~$P$ and $E(P)$ for the
edge set $\{\{v_i,v_{i+1}\} : 1\leq i\leq \ell\}$ of~$P$ and identify~$P$
with the subgraph of $G$ with vertex set $V(P)$ and edge set
$E(P)$. We say that the path~$P$ \emph{connects} its \emph{endpoints}
$v_1,v_{\ell+1}$, whereas $v_2,\ldots, v_\ell$ are the \emph{internal
  vertices} of $P$. 
Two vertices $u,v\in V(G)$ are \emph{connected} if there is a
path in $G$ with endpoints $u,v$.  The {\em{distance}} $\dist(u,v)$
between two connected vertices $u,v$ is the minimum length of a path
connecting $u$ and $v$; if $u,v$ are not connected, we put
$\dist(u,v)=\infty$. The \emph{radius} of~$G$ is
$\min_{u\in V(G)}\max_{v\in V(G)}\dist(u,v)$. 
A graph $G$ is $c$-\emph{degenerate} if every subgraph
$H\subseteq G$ has a vertex of degree at most $c$. For a positive
integer~$t$, we write $K_t$ for the complete graph on $t$ vertices.

For a vertex $v$ in graph $G$, the set of all
neighbors of a vertex $v$ in $G$ is denoted by $N^G(v)$, and the set
of all vertices at distance at most $r$ from $v$ (including $v$) is denoted 
by $N^G_r[v]$. We may omit the superscript if the graph is
clear from the context. 
The \emph{$r$-neighborhood complexity} of a subset $A\subseteq V(G)$ in
$G$ is defined as \[\nu_r(G, A)=|\{N_r[v]\cap A \colon v\in V(G)\}|.\]

\paragraph*{Nowhere dense classes of graphs.}
Let $G$ be a graph and let $r\in \N$. A graph $H$ with vertex set
$\{v_1,\ldots, v_n\}$ is a \emph{depth-$r$ minor} of~$G$, written
$H\minor_r G$, if there are connected and pairwise vertex disjoint
subgraphs $H_1,\ldots, H_n\subseteq G$, each of radius at most $r$,
such that if $\{v_i,v_j\}\in E(H)$, then there are $w_i\in V(H_i)$ and
$w_j\in V(H_j)$ such that $\{w_i,w_j\}\in E(G)$.


\begin{definition}
  A class $\CCC$ of graphs is \emph{nowhere dense} if there exists a
  function $f\colon \N\rightarrow \N$ such that $K_{f(r)}\not\minor_r G$ for
  all $r\in\N$ and for all $G\in \CCC$.
\end{definition}

Nowhere denseness admits several equivalent definitions, see~\cite{sparsity} for a wider discussion.
We next recall those that will be used in this paper.

\paragraph*{Edge density of shallow minors.}
For a graph $H$, the {\em{edge density}} of $H$ is the ratio $\frac{|E(H)|}{|V(H)|}$.
The {\em{grad}}, or {\em{greatest reduced average density}}, of a graph $G$ at depth $r\in \N$, is defined as:
\[\grad_r(G)=\max_{H\minor_r G}\frac{|E(H)|}{|V(H)|}.\]
For a class of graphs $\CCC$, let $\grad_r(\CCC)=\sup_{G\in \CCC} \grad_r(G)$.
Classes of bounded expansion are exactly those, for which $\grad_r(\CCC)$ is finite for each $r$.
In nowhere dense classes the edge density of shallow minors is not necessarily bounded by a constant, but it is small, as explained in the next result.

\begin{theorem}[\cite{dvorak2007asymptotical, nevsetvril2011nowhere}]\label{lem:gradbound}
Let $\CCC$ be a nowhere dense class of graphs. There is a function $\fgrad(r,\epsilon)$ such that for every $r\in \N$, $\epsilon>0$, graph $G\in \CCC$, and its $r$-shallow minor $H\minor_r G$,
it holds that $\frac{|E(H)|}{|V(H)|}\leq \fgrad(r,\epsilon) \cdot |V(H)|^\epsilon$.
\end{theorem}

\paragraph*{Weak coloring mumbers.}
For a graph $G$, by $\Pi(G)$ we denote the set of all linear orders
of $V(G)$. For $L\in\Pi(G)$, we write $u<_L v$ if $u$ is smaller than
$v$ in $L$, and $u\le_L v$ if $u<_L v$ or $u=v$.  For $u,v\in V(G)$ and
any $r\geq 0$, we say that~$u$ is \emph{weakly $r$-reachable} from~$v$
with respect to~$L$, if there is a path $P$ of length at most $r$ connecting $u$ and $v$ such that $u$ is 
the smallest among
the vertices of $P$ with respect to~$L$. By $\WReach_r[G,L,v]$ we
denote the set of vertices that are weakly $r$-reachable from~$v$ with
respect to~$L$. For any subset $A\subseteq V(G)$, we let
$\WReach_r[G,L,A] = \bigcup_{v\in A} \WReach_r[G,L,v]$.  The
\emph{weak $r$-coloring number $\wcol_r(G)$} of $G$ is defined as follows:
\begin{eqnarray*}
  \wcol_r(G)& = & \min_{L\in\Pi(G)}\:\max_{v\in V(G)}\:
                   \bigl|\WReach_r[G,L,v]\bigr|.\\
\end{eqnarray*}
The weak coloring numbers were introduced by Kierstead and
Yang~\cite{kierstead2003orders} in the context of coloring and
marking games on graphs. As proved by Zhu \cite{zhu2009coloring},
they can be used to characterize both bounded expansion and nowhere
dense classes of graphs. In particular, we use the following.

\begin{theorem}[\cite{zhu2009coloring}]\label{lem:wcolbound}
  Let $\CCC$ be a nowhere dense class of graphs.
  There is a function $f_{\wcol}(r,\epsilon)$ such that
  $\wcol_r(G')\leq f_{\wcol}(r,\epsilon) \cdot |V(G')|^\epsilon$ for
  every $r\in\N$, $\epsilon>0$ and $G'\subseteq G\in \CCC$.
\end{theorem}

The generalized coloring numbers play a key
role for Dvo\v{r}\'ak's approximation algorithm for $r$-dominating
sets~\cite{dvovrak2013constant} and also in the result on
neighborhood complexity by~\cite{reidl2016characterising}.  New
bounds for these numbers on restricted graph classes can be found
in~\cite{GroheKRSS15,KreutzerPRS16,siebertz16}.

\paragraph*{Quasi-wideness.}
A set $B\subseteq V(G)$ is called {\em{$r$-independent}} in $G$ if for all
distinct $u,v\in B$ we have $\dist_G(u,v)>r$.

\begin{definition}
  A class $\CCC$ of graphs is \emph{uniformly quasi-wide} if there are
  functions $N\colon \N\times\N\rightarrow \N$ and $s:\N\rightarrow \N$ such
  that for all $r,m\in \N$ and all subsets $A\subseteq V(G)$ for
  $G\in \CCC$ of size $\abs{A}\geq N(r,m)$ there is a set
  $S\subseteq V(G)$ of size $\abs{S}\leq s(r)$ and a set
  $B\subseteq A\setminus S$ of size $\abs{B}\geq m$ which is $r$-independent in
  $G-S$.  The functions $N$ and $s$ are called the \emph{margins} of
  the class~$\CCC$.
\end{definition}

It was shown by Ne\v{s}e\v{r}il and Ossona de
Mendez~\cite{nevsetvril2011nowhere} that a monotone class $\CCC$ of graphs is
nowhere dense if and only if it is uniformly quasi-wide. In particular,
every nowhere dense class is uniformly quasi-wide, even without the assumption of monotonicity.

For us it will be important that the margins $N$ and $s$ can be assumed to be
polynomial in~$r$ and that the sets $B$ and $S$ can be efficiently
computed. This was proved only recently by Kreutzer et
al.~\cite{siebertz2016polynomial}.  

\begin{theorem}[\cite{siebertz2016polynomial}]\label{thm:uqw}
  Let $\CCC$ be a nowhere dense class of graphs. For every $r\in \N$
  there exist constants~$p(r)$ and $s(r)$ such that
  for all $m\in \N$, all $G\in\CCC$ and all sets $A\subseteq V(G)$ of size at 
  least~$m^{p(r)}$, there is a set $S\subseteq V(G)$ of size at
  most $s(r)$ such that there is a set $B\subseteq A\setminus S$ of size at
  least~$m$ which is $r$-independent in $G-S$.

  Furthermore, if $K_c\not\minor_rG$ for all $G\in \CCC$, then
  $s(r)\leq c\cdot r$ and there is an algorithm, that given an $n$-vertex graph
  $G\in \CCC$, $r\in \N$, and $A\subseteq V(G)$ of size at least
  $m^{p(r)}$, computes a set $S$ of size at most $s(r)$ and an
  $r$-independent set $B\subseteq A\setminus S$ in $G-S$ of size at least $m$ in
  time $\Oof(r\cdot c\cdot |A|^{c+6}\cdot n^2)$, where $n=|V(G)|$.
\end{theorem}

The running time of the algorithm of~\cref{thm:uqw}
 is stated in~\cite{siebertz2016polynomial} only as 
$\Oof(r\cdot c\cdot n^{c+6})$.
However, it is easy to see that the algorithm actually runs in
time as stated above. The algorithm computes~$r$ times the truth 
value of $c$ first-order formulas of the form $\phi(x_1,\ldots, x_c)$, with $c$ free variables ranging over the parameter set~$A$.
That is, it decides for each tuple $(a_1,\ldots, a_c)\in A^c$ whether the formula $\phi(a_1,\ldots, a_c)$
is true in the graph. 
Each considered formula has quantifier rank one, so by brute-force
we can check all evaluations of the quantifier and we have to access
the adjacency list of each inspected vertex, giving a multiplicative factor $n^2$ in the running time. A more
careful analysis is certainly possible; however, we refrain from performing it as it is not crucial for the goals of this study.

\paragraph*{VC-dimension.}
Let $\FFF\subseteq 2^A$ be a family of
subsets of a set $A$. For a set $X\subseteq A$, we denote $X\cap \FFF=\{X\cap F : F\in \FFF\}$.
The set $X$ is \emph{shattered by $\FFF$} if $X\cap \FFF=2^X$.
The \emph{Vapnik-Chervonenkis dimension}, short \emph{VC-dimension},
of $\FFF$ is the maximum size of a set $X$ that is shattered by
$\FFF$. Note that if $X$ is shattered by $\FFF$, then also every
subset of $X$ is shattered by~$\FFF$.

The following theorem was first proved by Vapnik and
Chervonenkis~\cite{chervonenkis1971theory}, and rediscovered by
Sauer~\cite{sauer1972density} and
Shelah~\cite{shelah1972combinatorial}. It is often called the
Sauer-Shelah lemma in the literature.
\begin{theorem}[Sauer-Shelah Lemma]\label{thm:sauer_shelah}
  If $|A|\leq n$ and $\FFF\subseteq 2^A$ has VC-dimension $d$,
  then
  \[|\FFF|\leq \sum_{i=0}^{d}\binom{n}{i}\in \Oof(n^d).\]
\end{theorem}
Note that in the interesting cases $d\geq 2, n\geq 2$ it holds that
$\sum_{i=0}^{d}\binom{n}{i}\leq n^d$, and in general it holds that $\sum_{i=0}^{d}\binom{n}{i}\leq 2\cdot n^d$.

For a graph $G$, the VC-dimension of $G$ is defined as the VC-dimension
of the set family $\{N[v]\colon v\in V(G)\}$ over the set $V(G)$.

\paragraph*{VC-dimension and nowhere denseness.} Adler and
Adler~\cite{adler2014interpreting} have proved that any nowhere dense class $\CCC$
of graphs is {\em{stable}}, which in particular implies that any class of structures
obtained from $\CCC$ by means of a first-order interpretation has VC-dimension
bounded by a constant depending only on $\CCC$ and the interpretation.
In particular, the following is an immediate corollary of the results of Adler and Adler~\cite{adler2014interpreting}.

\begin{theorem}[\cite{adler2014interpreting}]\label{thm:adler}
  Let $\CCC$ be a nowhere dense class of graphs and let $\phi(x,y)$ be
  a first-order formula over the signature of graphs,
  such that for all $G \in \CCC$ and $u,v\in
  V(G)$ it holds that $G\models\phi(u,v)$ if and only if $G\models\phi(v,u)$. 
  For $G\in \CCC$, let $G_\phi$ be the graph with
  vertex set $V(G_\phi)=V(G)$ and edge set $E(G_\phi)=\{\{u,v\} \colon
  G\models\phi(u,v)\}$. Then there is an integer $c$ depending only on
  $\CCC$ and $\phi$ such that $G_\phi$ has VC-dimension at most $c$.
\end{theorem}

By applying \cref{thm:adler} to first-order formulas expressing the properties $\dist(u,v)\leq r$ and $\dist(u,v)=r$,
we immediately obtain the following.

\begin{corollary}\label{thm:adleradler}
  Let $\CCC$ be a nowhere dense class of graphs and let $r\in\N$. For
  $G\in \CCC$, let $G_{=r}$ be the
  graph with the same vertex set as $G$ and an edge
  $\{u,v\}\in E(G_{=r})$
  if and only if $\dist_G(u,v)=r$.
  Define the graph $G_{\leq r}$ in the same manner, but putting an edge $\{u,v\}$ 
  into $E(G_{\leq r})$ if and only if $\dist_G(u,v)\leq r$. 
  Then there is an integer $c(r)$ such that both $G_{=r}$ and $G_{\leq r}$ have VC-dimension at most~$c(r)$ for every $G\in \CCC$.
\end{corollary}

By combining \cref{thm:adleradler} with the Sauer-Shelah Lemma (\cref{thm:sauer_shelah}) we infer the following.

\begin{corollary}\label{lem:numberofneighborhoods}
  Let $\CCC$ be a nowhere dense class of graphs and let $r\in \N$.  Then
  $\nu_r(G,A)\leq |A|^{c(r)}$ for every graph $G\in \CCC$ and 
  $A\subseteq V(G)$ with $|A|\geq 2$, where $c(r)$ is the constant given by~\cref{thm:adleradler}.
\end{corollary}

Thus, a polynomial bound on the neighborhood complexity for any nowhere dense class, and, in fact, for any stable class, already follows from known tools.
Our goal in the next section will be to show that with the assumption of nowhere denseness we can prove an almost linear bound, as described in \cref{thm:main-neighborhood-complexity}.

\paragraph*{Distance profiles.}
In our reasoning we will need a somewhat finer view of the neighborhood complexity.
More precisely, we would like to partition the vertices of the graph not only with respect to their $r$-neighborhood in a fixed set $A$, but also with respect to what are the exact distances
of the elements of this $r$-neighborhood from the considered vertex. With this intuition in mind, we introduce the notion of a {\em{distance profile}}.

Let $G$ be a graph and let $A\subseteq V(G)$ be a subset of its vertices.
For a vertex $u\in V(G)$, the {\em{$r$-distance profile}} of $u$, denoted $\pi^G_r[u,A]$, is
a function mapping vertices of $A$ to $\{0,1,\ldots,r,\infty\}$ defined as follows:
\[ \pi^G_r[u,A](v)=\begin{cases}\dist_G(u,v) & \quad \textrm{if }\dist_G(u,v)\leq r,\\ \infty & \quad \textrm{otherwise.}\end{cases} \]
We say that a function $f\colon A\to \{0,1,\ldots,r,\infty\}$ is \emph{realized as an $r$-distance profile on $A$} if there is
$u\in V(G)$ such that $f=\pi^G_r[u,A]$. We may drop the superscript if the graph is clear from the context.

Similarly to the neighborhood complexity, we define the {\em{distance profile complexity}} of a vertex subset $A\subseteq V(G)$ in a graph $G$, denoted $\profnum_r(G,A)$, as the
number of different functions realized as $r$-distance profiles on $A$ in $G$. Clearly it always holds that $\nu_r(G,A)\leq \profnum_r(G,A)$, 
thus~\cref{thm:main-neighborhood-complexity} will follow directly from the following result, which will be proved in the next section. 

\begin{theorem}\label{thm:main-profiles}
  Let $\CCC$ be a nowhere dense class of graphs. Then there is 
  a function $\fnei(r,\epsilon)$ such that for every $r\in \N$, 
  $\epsilon>0$, graph $G\in \CCC$, and vertex subset $A\subseteq V(G)$, 
  it holds that $\profnum_r(G,A)\leq \fnei(r,\epsilon)\cdot |A|^{1+\epsilon}$.
\end{theorem}

Let us observe that the polynomial bound of \cref{lem:numberofneighborhoods} carries over to distance profiles.
 
\begin{lemma}\label{crl:number-of-profiles}
Let $\CCC$ be a nowhere dense class of graphs. Then there is an integer $d(r)$ such that for every $r\in \N$, 
  $\epsilon>0$, graph $G\in \CCC$, and vertex subset $A\subseteq V(G)$ with $|A|\geq 2$, 
  it holds that $\profnum_r(G,A)\leq |A|^{d(r)}$.
\end{lemma} 
\begin{proof}
Observe that for any vertex $u\in V(G)$, the $r$-distance profile $\pi^G_r[u,A]$ is uniquely determined by the tuple
\[(N^G_0[u]\cap A,N^G_1[u]\cap A,\ldots,N^G_{r-1}[u]\cap A,N^G_r[u]\cap A).\]
This is because for any $v\in A$, the value $\pi^G_r[u,A][v]$ is equal to the smallest integer $i\leq r$ such that $v\in N^G_i(u)$, or $\infty$ if no such integer exists.
By~\cref{lem:numberofneighborhoods}, we have $|A|^{c(i)}$ possibilities for the $i$-th entry of this tuple, hence we can set $d(r)=c(0)+c(1)+\ldots+c(r)$.
\end{proof}

\section{Neighborhood complexity of nowhere dense classes}
\label{sec:neighborhood-complexity}

Reidl et al.~\cite{reidl2016characterising} gave a linear bound on the neighborhood complexity for
all subsets of vertices of graphs from a class of bounded expansion; see Theorem 3 in~\cite{reidl2016characterising}.
More precisely, they show that for every graph $G$ and every
$A\subseteq V(G)$, it holds that
\[\nu_r(G,A) \le \frac{1}{2}(2r+2)^{\wcol_{2r}(G)}\cdot \wcol_{2r}(G)\cdot |A|+1.\] 
Note that if $G$ belongs to some fixed class $\CCC$ of bounded expansion, then $\wcol_{2r}(G)$ is bounded by a constant depending only on
the class $\CCC$ and $r$.
However, assuming only that $\CCC$ is nowhere dense,
by \cref{lem:wcolbound} we have that for all
$r\geq 1$ and $\epsilon>0$ the number $\wcol_{2r}(G)$ is bounded only
by $\fwcol(2r,\epsilon) \cdot |V(G)|^\epsilon$, which is not a constant.
In particular, since $\wcol_{2r}(G)$ appears in the exponent of the bound of Reidl et al.,
we do not obtain even a polynomial upper bound.

In this section we prove~\cref{thm:main-profiles}, which directly implies~\cref{thm:main-neighborhood-complexity}, as explained in the previous section.
Our approach is to carefully analyze the proof of Reidl et al.~\cite{reidl2016characterising}, and
to fix parts that break down in the nowhere dense setting using tools derived, essentially, from the stability of nowhere dense classes. 

We first prove the following auxiliary lemma. We believe it may be of
independent interest, as it seems very useful for the analysis of weak coloring numbers in the nowhere dense setting.

\pagebreak

\begin{lemma}\label{lem:bd-vcdimension}
  Let $\CCC$ be a nowhere dense class of graphs and let $G \in \CCC$. For
  $r \geq 0$ and a linear order $L \in \Pi(G)$, let
  \[
  \WWW_{r,L} = \{ \WReach_r[G,L,v]\ \colon\ v \in V(G) \}.
  \]
  Then there is a constant $x(r)$, depending only on $\CCC$ and $r$ (and
  not on $G$ and $L$), such that $\WWW_{r,L}$ has VC-dimension at most
  $x(r)$.
\end{lemma}
\begin{proof}
  Since $\CCC$ is nowhere dense, according to~\cref{lem:wcolbound}, it
  is uniformly quasi-wide, say with margins $N$ and $s$. We fix a
  number $m$ to be determined later, depending only on $r$ and $\CCC$.
  Let $x = x(r) = N(2r,m)$ and $s = s(r)$. Assume towards a
  contradiction that there is a
  set $A \subseteq V(G)$ of size $x$ which is shattered by
  $\WWW_{r,L}$. Fix sets $S \subseteq V(G)$ and $B \subseteq A\setminus S$ such that $|B| = m$,
  $|S| \leq s$, and $B$ is $2r$-independent in $G - S$.
  We will treat $L$ also as a linear order on the vertex set of $G-S$.
  
  As a subset of $A$, the set $B$ is also shattered by $\WWW_{r,L}$.
  That is, for every $X \subseteq B$ there is a vertex $v_X \in V(G)$ such that $X = \WReach_r[G,L,v_X] \cap B$.
  Note that since $B$ is $2r$-independent in $G - S$, so is $X$, and we have
  \[
  |\WReach_r[G-S,L,v_X] \cap B| \leq 1.
  \]
  
  For a vertex $\sigma \in S$, number $\rho \in \{ 0,\ldots,r-1 \}$, and vertex $v\in V(G)$,
  let us consider the set $\PPP_{v,\sigma,\rho}$ of all paths of length (exactly) $\rho$ that connect $\sigma$ and $v$.
  We define $b_{\sigma,\rho}(v)$ to be the largest (with respect to $L$) vertex $b\in B$, for which there exists a path $P\in \PPP_{v,\sigma,\rho}$
  such that every vertex on $P$ is strictly larger than $b$ with respect to $L$. If no such vertex in $B$ exists, we put $b_{\sigma,\rho}(v)=\bot$. 
  
  The \emph{signature} of a vertex $v \in V(G)$ is defined as
  \[
  \chi(v) =  \big(b_{\sigma,\rho}(v)\big)_{\sigma \in S,\, 0 \leq \rho \leq r-1}.
  \]
  It now follows that the number of possible signatures is small.
  \begin{claim}\label{cl:num-signatures}
  The set $\{\chi(v)\colon v\in V(G)\}$ has size at most $(m+1)^{r\cdot s}$.
  \end{claim}
  \begin{clproof}
  For each $\sigma\in S$ and $\rho\in\{0,\ldots,r-1\}$, there are $|B|+1=m+1$ possible values for $b_{\sigma,\rho}(v)$, while
  the number of relevant pairs $(\sigma,\rho)$ is at most $r\cdot s$.
  \end{clproof}
  
  The next claim intuitively shows that for a vertex $v$, the signature of $v$ plus the set of vertices of $B$ weakly reachable from $v$ in $G-S$ provides enough information to deduce 
  precisely the set of vertices of $B$ weakly reachable from $v$ in $G$.
  
  \begin{claim}\label{cl:signatures-enough}
  Suppose $v,w\in V(G)$ are such that 
  \[\chi(v) = \chi(w)\quad \textrm{ and }\quad \WReach_r[G-S,L,v]\cap B = \WReach_r[G-S,L,w]\cap B.\]
  Then $\WReach_r[G,L,v]\cap B = \WReach_r[G,L,w]\cap B$.
  \end{claim}
  \begin{clproof}
  By symmetry it suffices to show that for any $b \in B$, if
  $b \in \WReach_r[G,L,v]$, then $b \in \WReach_r[G,L,w]$. For every such
  $b$, there is a path $P$ of length a most $r$ that connects $v$ and $b$, on which, moreover, $b$ is the
  $L$-minimal vertex. If this path avoids $S$, then
  \[b \in \WReach_r[G-S,L,v] = \WReach_r[G-S,L,w]\subseteq \WReach_r[G,L,w],\]
  and we are done. Otherwise, let $\sigma$ be the first vertex in $S$ on $P$ (i.e., the closest to
  $v$); obviously $\sigma\neq b$, since $b\notin S$. 
  Let $\sigma$ be the $\rho$-th vertex on $P$ (starting from $v$), where we have $\rho\in \{0,1,\ldots,r-1\}$. 
  Let $P_1$ be the prefix of $P$ from $v$ to $\sigma$, and $P_2$ be the suffix of $P$ from $\sigma$ to~$b$. Since $b$ is the
  $L$-minimal vertex on $P$, and $\sigma\neq b$, all vertices on $P_1$ are greater 
  than $b$ with respect to $L$. Therefore $b_{\sigma,\rho}(v)\neq \bot$ and $b \leq_L b_{\sigma,\rho}(v)$. 
  Since $\chi(v) = \chi(w)$, we have that $b_{\sigma,\rho}(v)=b_{\sigma,\rho}(w)$. 
  Consequently, there is a path $Q$ of length $\rho$ connecting $w$ and $\sigma$ such that all 
  vertices of $Q$ are greater than $b$ with respect to $L$. 
  By concatenating $Q$ and $P_2$ we obtain a path of length at most $r$ connecting $w$ and $b$, such that $b$ is the $L$-minimal vertex on this path.
  This implies that $b \in \WReach_r[G,L,w]$.
  \end{clproof}
  
  By~\cref{cl:num-signatures} there are at most $(m+1)^{r\cdot s}$ possible signatures, while we argued that 
  the intersection $\WReach_r[G-S,L,v]\cap B$ is always of size at most $1$, hence there are at most $(m+1)$ possibilities for it.
  Thus, by~\cref{cl:signatures-enough} we conclude that only at most $(m+1)^{r\cdot s} \cdot (m+1)$  subsets
  of~$B$ are realized as $B \cap W$ for some $W \in \WWW_{r,L}$. 
  To obtain a contradiction with $B$ being shattered by $\WWW_{r,L}$, it suffices to select $m$ so that
  \[(m+1)^{r\cdot s+1}<2^m.\]
  Since $s=s(r)$ is a constant depending on $\CCC$ and~$r$ only, we may choose $m$ depending on $\CCC$ and~$r$ so that the above inequality holds. 
\end{proof}


With~\cref{lem:bd-vcdimension} in hand, we now are ready to prove~\cref{thm:main-profiles}. 

\setcounter{theorem}{7}

\begin{theorem}[restated]
Let $\CCC$ be a nowhere dense class of graphs. Then there is 
  a function $\fnei(r,\epsilon)$ such that for every $r\in \N$, 
  $\epsilon>0$, graph $G\in \CCC$, and vertex subset $A\subseteq V(G)$, 
  it holds that $\profnum_r(G,A)\leq \fnei(r,\epsilon)\cdot |A|^{1+\epsilon}$.
\end{theorem}
\setcounter{theorem}{6}
\begin{proof}
  Fix $r\geq 1$ and $\epsilon>0$. 
  In the proof we will use small constants $\epsilon_1,\epsilon_2, \epsilon_3,\epsilon_4>0$, which will be determined in the course of the
  proof.
  If $|A|\leq r$, then the number of different $r$-distance profiles on $A$ is at most $(r+2)^r$, which is bounded by a function of $r$ only.
  Hence, throughout the proof we can assume without loss of generality that $|A|>r$, in particular $|A|\geq 2$.
  
  The first step is to reduce the problem to the case when the size of the graph is bounded polynomially in $|A|$.
  According to~\cref{crl:number-of-profiles}, there is an integer
  $d(r)$ such that there are at most $|A|^{d(r)}$ different $r$-distance profiles on $A$. 
  We therefore classify the
  elements of $V(G)$ according to their $r$-distance profiles on $A$,
  that is, let define an equivalence relation $\sim$ on $V(G)$ as follows:
  \[v\sim w \quad \textrm{if and only if}\quad \pi^G_r[v,A] = \pi^G_r[w,A].\] 
  Construct a set $A'$ by taking $A$ and, for each equivalence class $\kappa$ of $\sim$, adding an arbitrary element $v_{\kappa}$ to $A'$.
  Then, construct a set $A''$ by starting with $A'$, and, for each distinct $u,v\in A''$, performing the following operation:
  if $\dist_G(u,v)\leq r$, then add the vertex set of any shortest path between $u$ and $v$ to $A''$. Finally, let $G'=G[A'']$.
  
  \begin{claim}\label{cl:size-Gp}
  The following holds: $|A''|\leq |A|^{2\cdot d(r)+3}$.
  \end{claim}
  \begin{clproof}
  By~\cref{crl:number-of-profiles}, the equivalence relation $\sim$ has at most $|A|^{d(r)}$ equivalence classes, hence $|A'|\leq |A|+|A|^{d(r)}$.
  When constructing $|A''|$ we consider at most $|A'|^2$ pairs of vertices from $A$, and for each of them we add at most $r-1$ vertices to $A''$.
  It follows that
  \[|A''|\leq |A'|+(r-1)|A'|^2 \leq r|A'|^2\leq r(|A|+|A|^{d(r)})^2\leq 4r|A|^{2\cdot d(r)}\leq |A|^{2\cdot d(r)+3}.\]
  \end{clproof}
  
  \begin{claim}\label{cl:preserve-Gp}
  The following holds: $\profnum_r(G',A)\geq \profnum_r(G,A)$.
  \end{claim}
  \begin{clproof}
  Observe that for any $u\in A,v\in A'$, if $\dist_G(u,v)\leq r$, then $\dist_{G'}(u,v)=\dist_{G}(u,v)$; this follows immediately from the inclusion of shortest paths in the construction of $A''$.
  Therefore, we infer that for each equivalence class $\kappa$ of $\sim$, the $r$-distance profile of $v_\kappa$ on $A$ in $G'$ is the same as in $G$, because $v_\kappa\in A'$ and $A\subseteq A'$.
  It follows that every $r$-distance profile on $A$ that is realized in $G$ is also realized in $G'$.
  \end{clproof}
  
  Intuitively,~\cref{cl:preserve-Gp} shows that we can restrict our attention to graph $G'$ instead of $G$, since by doing this we do not lose any $r$-distance profiles realized on $A$.
  The gain from this step is that the size of $G'$ is bounded polynomially in terms of $|A|$ (\cref{cl:size-Gp}), hence we can use better bounds on the weak coloring numbers, as explained next.

  According to~\cref{lem:wcolbound}, there is a function $\fwcol$
  such that
  \[\wcol_{2r}(G')\leq \fwcol(2r,\epsilon_4) \cdot
  |A''|^{\epsilon_4}=\fwcol(2r,\epsilon_4) \cdot
  |A|^{(2\cdot d(r)+3)\cdot \epsilon_4}=\fwcol(2r,\epsilon_4) \cdot
  |A|^{\epsilon_3},\] where $\epsilon_4=\epsilon_3/(2\cdot d(r)+3)$.
  Let $L$ be a linear order of $V(G')$ witnessing
  $\wcol_{2r}(G')\leq \fwcol(2r,\epsilon_4)\cdot|A|^{\epsilon_3}$;
  that is, for each $v\in V(G')$ we have $|\WReach_{2r}[G',L,v]|\leq \fwcol(2r,\epsilon_4)\cdot|A|^{\epsilon_3}$.
  For each $v\in V(G')$, let us define the following set:
  \[
  Y[v] = \WReach_r[G', L, v] \cap \WReach_r[G',L,A].
  \]
  In other words, $Y[v]$ comprises all vertices that are weakly $r$-reachable both
  from $v$ and from some vertex of $A$. Since $Y[v]\subseteq \WReach_r[G', L, v]$ for each $v\in V(G')$,  we have
  \[|Y[v]| \leq |\WReach_r[G', L, v]|\leq \fwcol(2r,\epsilon_4)\cdot |A|^{\epsilon_3}. \]
    Furthermore, as for each $v\in V(G')$ we have
    $Y[v]\subseteq \WReach_r[G,L,A] = \bigcup_{w\in
      A}\WReach_r[G,L,w]$,
    \[ \Big|\bigcup_{v\in V(G')} Y[v]\Big| \leq  |A|\cdot \max_{w\in V(G')} |\WReach_r[G,L,w]| \leq
    \fwcol(2r,\varepsilon_4)\cdot|A|^{1+\varepsilon_3}.
  \]
  
  We now classify the vertices $v\in V(G')$ according to their
  distance profiles $\pi^{G'}_r[v, Y[v]]$. More precisely,
  let $\equiv$ be the equivalence relation on $V(G')$ defined as follows:
  \[v\equiv w\quad\textrm{if and only if}\quad Y[v] = Y[w]\textrm{ and }\pi^{G'}_r[v,Y[v]]=\pi^{G'}_r[w,Y[w]].\]
  We next show that the equivalence relation $\equiv$ refines the standard partitioning according to $r$-distance profiles on $A$.
  
  \begin{claim}\label{cl:equiv-refines}
  For every $v,w\in V(G)$, if $v\equiv w$, then $\pi^{G'}_r[v,A]=\pi^{G'}_r[w,A]$.
  \end{claim}
  \begin{clproof}
  By symmetry, to show that $\pi^{G'}_r[v,A]=\pi^{G'}_r[w,A]$ it suffices to prove 
  that for each $x\in A$, if $\dist_{G'}(v,x)\leq r$, then $\dist_{G'}(w,x)\leq \dist_{G'}(v,x)$.
  Let $P$ be a shortest path from $v$ to $x$ in $G'$, say of length $\ell\leq r$.
  Let $y$ be the smallest vertex on $P$ with respect to $L$. 
  Then $P$ certifies that $y$ is both weakly $r$-reachable from $v$ and from $x\in A$, hence $y \in Y[v]$.
  Since $Y[v]=Y[w]$ and $\pi^{G'}_r[v,Y[v]]=\pi^{G'}_r[w,Y[w]]$ by assumption, and the prefix of $P$ from $v$ to $y$ certifies that $\dist_{G'}(v,y)\leq r$, we have that 
  $\dist_{G'}(v,y)=\dist_{G'}(w,y)$.
  Since $P$ was a shortest path, also the prefix of $P$ from $v$ to $y$ is a shortest path between $v$ and $y$, and the suffix of $P$ from $y$ to $x$ is a shortest path between $y$ and $x$.
  Concluding,
  \[\dist_{G'}(w,x)\leq \dist_{G'}(w,y)+\dist_{G'}(y,x)=\dist_{G'}(v,y)+\dist_{G'}(y,x)=\ell.\]
  \end{clproof}
  
  \cref{cl:equiv-refines} suggests the following approach to bounding $\profnum_r(G',A)$: first give an upper bound on the number of possible sets of the form $Y[v]$, and then
  for each such set, bound the number of $r$-distance profiles on it. We deal with the second part first, as it essentially follows from~\cref{crl:number-of-profiles}.
  Let us set $\epsilon_2=\epsilon_3\cdot d(r)$, where $d(r)$ is the constant given by~\cref{crl:number-of-profiles}.
  
  \begin{claim}\label{cl:prof-on-one}
  There is a function $g(r,\varepsilon_4)$ such that for each $v\in V(G)$, 
  we have $\profnum_r(G,Y[v])\leq g(r,\epsilon_4)\cdot |A|^{\epsilon_2}$.
  \end{claim}
  \begin{clproof}
  Fix some $Z=Y[v]$ for $v\in V(G)$. If $|Z|=1$, then there are
  obviously at most $r+1$ distance profiles on $Z$. Otherwise,
  according to~\cref{crl:number-of-profiles}, there is a constant
  $d(r)$ such that $\profnum_r(G,Z)\leq |Z|^{d(r)}$. 
  Since $|Z|=|Y[v]|\leq \fwcol(2r,\epsilon_4)\cdot |A|^{\epsilon_3}$, we obtain that the
  total number of profiles on $Z$ is bounded by
  \[(r+1+\fwcol(2r,\epsilon_4))^{d(r)}\cdot |A|^{\epsilon_3\cdot d(r)} = (r+1+\fwcol(2r,\epsilon_4))^{d(r)})\cdot |A|^{\epsilon_2}.\]
  Hence we can take $g(r,\varepsilon_4)=(r+1+\fwcol(2r,\varepsilon_4))^{d(r)}$.
  \end{clproof}
  
  It remains to give an upper bound on the number of distinct sets $Y[v]$.
  Let us set $\epsilon_1 = \epsilon_3\cdot x(r)$, where $x(r)$ is the constant given by~\cref{lem:bd-vcdimension}.
  
  \begin{claim}\label{cl:num-Y}
  The following holds: $\left| \{ Y[v] \colon v\in V(G') \} \right|\leq 1+2\cdot \fwcol(2r,\epsilon_4)^{x(r)+1} \cdot |A|^{1+\epsilon_1+\epsilon_3}.$
  \end{claim}
  \begin{clproof}
  Let $\YYY=\{Y[v] \colon v\in V(G')\}\setminus \{\emptyset\}$ be the family of all non-empty sets of the form $Y[v]$ for $v\in V(G')$; it suffices to show that
  $|\YYY|\leq \fwcol(2r,\epsilon_4)^{x(r)+1} \cdot |A|^{1+\epsilon_1+\epsilon_3}$.
  Define mapping $\gamma\colon \YYY\rightarrow V(G')$ as follows: for $Z\in \YYY$, $\gamma(Z)$ is the largest element of $Z$ with respect to $L$.

  Take any $v\in V(G')$ with $Y[v]\neq \emptyset$, and recall that every vertex in $Y[v]$ is weakly $r$-reachable from $v$. Observe that every vertex $w\in Y[v]$ is weakly
  $2r$-reachable from $\gamma(Y[v])$. To see this, concatenate the two paths of length at most $r$ that certify that $w\in \WReach_r[G',L,v]$ and $\gamma(Y[v])\in \WReach_r[G',L,v]$,
  and note that this path of length at most $2r$ certifies that $w\in \WReach_{2r}[G',L,\gamma(Y[v])]$.
  Consequently, for every $y\in \gamma(\YYY)$, we have
  \[
  \bigcup\gamma^{-1}(y)\subseteq \WReach_{2r}[G',L,y].
  \]
  This implies that the union of all sets of $\YYY$ that choose the same $y$
  via $\gamma$ has size at most $\wcol_{2r}(G')$.

  Fix any $y\in \gamma(\YYY)$ and denote $S_y=\bigcup\gamma^{-1}(y)$
  Let us count how many distinct subsets of $S_y$ belong to $\YYY$.
  Every such $Z=Y[v]$, as a subset of $S_y$, satisfies
  $Y[v]\cap S_y = Y[v] = \WReach_r[G'L,v] \cap \WReach_r[G',L,A]$. As
  for all $v$ the set $\WReach_r[G',L,A]$ is the same, this means that
  the number of different $Y[v]\in \YYY$ that are mapped to a fixed $y$ is not larger 
  than the number of different sets $\WReach_r[G',L,v]\cap S_y$, for $v\in V(G')$.
  
  By Lemma~\cref{lem:bd-vcdimension}, the set
  $\WWW_{r,L} = \{ \WReach_r[G',L,v] \suchthat v \in V(G') \}$ has
  VC-dimension at most $x(r)$, and so has the subfamily
  $\{S_y\cap \WReach_r[G',L,v] \suchthat v \in V(G')\}$. 
  Hence, by the Sauer-Shelah Lemma (\cref{thm:sauer_shelah}), we infer that
  \[
  |\{S_y\cap \WReach_r[G',L,v] \colon v\in V(G')\}|\leq 2\cdot |S_y|^{x(r)}.
  \]
  Finally, we observe that $\gamma(\YYY)\subseteq \bigcup \YYY$ and recall that $|\bigcup \YYY|\leq \fwcol(2r,\varepsilon_4)\cdot|A|^{1+\epsilon_3}$, hence
  \begin{eqnarray*}
  |\YYY| & \leq & \sum_{y\in \gamma(\YYY)} |\{S_y\cap \WReach_r[G',L,v] \colon v\in V(G')\}| \\
  &\leq & 2\cdot \sum_{y\in \gamma(\YYY)} |S_y|^{x(r)} \leq 2\cdot |\gamma(\YYY)|\cdot (\wcol_{2r}(G'))^{x(r)} \\
         & \leq & 2\cdot |\bigcup \YYY|\cdot (\fwcol(2r,\epsilon_4)\cdot |A|^{\epsilon_3})^{x(r)} \leq  2\cdot \fwcol(2r,\varepsilon_4)\cdot|A|^{1+\epsilon_3}\cdot \fwcol(2r,\epsilon_4)^{x(r)+1}\cdot |A|^{\epsilon_1}\\
         & = & 2\cdot \fwcol(2r,\epsilon_4)^{x(r)+1}\cdot |A|^{1+\varepsilon_1+\epsilon_3}.
  \end{eqnarray*}
  \end{clproof}

  By combining~\cref{cl:preserve-Gp},~\cref{cl:equiv-refines},~\cref{cl:prof-on-one}, and~\cref{cl:num-Y}, we conclude that
  \begin{eqnarray*}
    \profnum_r(G,A) & \leq &  \profnum_r(G',A) \le \indx(\equiv) \\
    & \leq & \left| \{ Y[v] \mid v\in V(G') \} \right|\cdot g(r,\epsilon_4)\cdot |A|^{\epsilon_2}\\
    & \leq & (1+2\cdot \fwcol(2r,\epsilon_4)^{x(r)+1} \cdot |A|^{1+\epsilon_1+\epsilon_3})\cdot g(r,\epsilon_4)\cdot |A|^{\epsilon_2}\\
    & \leq & 3\cdot \fwcol(2r,\epsilon_4)^{x(r)+1}\cdot g(r,\epsilon_4)\cdot |A|^{1+\epsilon_1+\epsilon_2+\epsilon_3}.
  \end{eqnarray*}
  Thus, if we fix $\epsilon_4>0$ so that $\epsilon_1+\epsilon_2+\epsilon_3<\epsilon$, then we can set $\fnei(r,\epsilon)=3\cdot \fwcol(2r,\epsilon_4)^{x(r)+1}\cdot g(r,\epsilon_4)$.
\end{proof}

\section{Kernelization for distance-$r$ dominating sets}\label{sec:kernel}

In this section we use the neighborhoods complexity tools to prove~\cref{thm:main-kernel}.
Our proof will rely on the approach of Drange et al.~\cite{drange2016kernelization}, and therefore we need to first explain some tools borrowed from there,
in particular we recall the notion of $r$-projection and introduce $r$-projection profiles.
Throughout the section we fix a nowhere dense class $\CCC$.
Whenever we say that the running time of some algorithm on a graph $G$ is {\em{polynomial}}, we mean that it is of the form $\Oof((|V(G)|+|E(G)|)^\cst)$,
where $\cst$ is a universal constant that is independent of $\CCC$, $r$, $\epsilon$, or any other constants defined in the context. However,
the constants hidden in the $\Oof(\cdot)$-notation may depend on $\CCC$, $r$, and $\epsilon$.

\paragraph*{Projections and projection profiles.}
Let $G\in \CCC$ be a graph and let $A\subseteq V(G)$ be a subset of vertices. For vertices $v\in A$ and $u\in V(G)\setminus A$, a path $P$ connecting $u$ and $v$ is called {\em{$A$-avoiding}}
if all its vertices apart from $v$ do not belong to $A$. For a positive integer $r$, the {\em{$r$-projection}} of any $u\in V(G)\setminus A$ on $A$, denoted $M^G_r(u,A)$ is the set of all vertices $v\in A$ that
can be connected to $u$ by an $A$-avoiding path of length at most $r$. The {\em{$r$-projection profile}} of a vertex $u\in V(G)\setminus A$ on $A$ is a function $\rho^G_r[u,A]$ mapping vertices of
$A$ to $\{0,1,\ldots,r,\infty\}$, defined as follows: for every $v\in A$, the value $\rho^G_r[u,A](v)$ is the length of a shortest $A$-avoiding path connecting $u$ and~$v$, and~$\infty$ in case this length
is larger than $r$. Similarly as for $r$-neighborhoods and $r$-distance profiles, we define 
\[\projnum_r(G,A)=|\{M_r^G(u,A)\colon u\in V(G)\setminus A\}|\quad\textrm{and}\quad \projprof_r(G,A)=|\{\rho_r^G[u,A]\colon u\in V(G)\setminus A\}|\]
to be the number of different $r$-projections and $r$-projection profiles realized on $A$, respectively. Clearly, again it always holds that $\projnum_r(G,A)\leq \projprof_r(G,A)$.

In their kernelization algorithm, Drange et al.~\cite{drange2016kernelization} use a linear bound on the number of different $r$-projections in any graph class of bounded expansion.
We now verify that, in the nowhere dense setting the near-linear bound on the number of different $r$-distance profiles can be used also to give a near-linear bound the number of different $r$-projection profiles.
The proof is based on the idea of creating a ``layered'' graph, which was also used by Drange et al.~\cite{drange2016kernelization}.

\begin{lemma}\label{lem:projection-complexity}
  Suppose $\CCC$ is a nowhere dense class of graphs. Then there is 
  a function $\fproj(r,\epsilon)$ such that for every $r\in \N$, 
  $\epsilon>0$, graph $G\in \CCC$, and vertex subset $A\subseteq V(G)$, 
  it holds that $\projprof_r(G,A)\leq \fproj(r,\epsilon)\cdot |A|^{1+\epsilon}$.
\end{lemma}
\begin{proof}
For a graph $G$ and a positive integer $c$, the {\em{$c$-blowup of $G$}}, denoted $G\bullet K_c$, is defined as follows\footnote{The $c$-blowup is a special case of a more general operation called 
{\em{lexicographic product}}, but we will not use this operation here.}.
The vertex set of $G\bullet K_c$ consists of pairs of the form $(u,i)$, where $u\in V(G)$ and $i\in \{1,2,\ldots,c\}$;
vertex $(u,i)$ is called the {\em{$i$-th}} copy of $u$.
The copies of every vertex $u\in V(G)$ are made into a clique, i.e., there is an edge $(u,i)(u,j)$ for all $u\in V(G)$ and $1\leq i<j\leq c$.
Moreover, whenever there is an edge $\{u,v\}\in V(G)$, then we make every copy of $u$ adjacent to every copy of~$v$, i.e., we put an edge $\{(u,i)(v,j)\}$ for all $1\leq i,j\leq c$.

For a class $\CCC$, by $\CCC^c$ we denote the closure under subgraphs of the class $\{G\bullet K_c\colon G\in \CCC\}$. 
It is well-known that if $\CCC$ is nowhere dense, then for any fixed positive integer $c$, the class $\CCC^c$ is also nowhere dense; see e.g.\ Chapter 4.7 of~\cite{sparsity}, or, simply observe that if $\CCC$ is uniformly quasi-wide, say with margins $s$ and $N$, then $\CCC^c$ is uniformly quasi-wide with margins $c\cdot s$ and $c\cdot N$.

Let us fix $r\in \N$ and $\epsilon>0$. Given $G\in \CCC$ and $A\subseteq V(G)$, we construct a graph $H$ as follows.
The vertex set of $H$ consists of pairs $(u,i)$, where $u\in V(G)$ and $i\in \{0,1,\ldots,r\}$.
Again, vertex $(u,i)$ is called the {\em{$i$-th}} copy of $u$, while the subset $V(G)\times \{i\}$ will be called the {\em{$i$-th layer}} of $H$.
The edge set of $H$ is defined as follows: for each edge $\{u,v\}\in
E(G)$ and each $i\in \{1,2,\ldots,r\}$, we put an edge $\{(u,i-1),(v,i)\}$, 
but only if $u\notin A$, and we put an edge $\{(v,i-1),(u,i)\}$, but only if $v\notin A$.
It is easy to see that $H$ is a subgraph of $G\bullet K_{r+1}$, the $(r+1)$-blowup of $G$, hence $H\in \CCC^{r+1}$.

Let $B=A\times \{0,1,\ldots,r\}$. We now prove that for every $u\in V(G)\setminus A$, the $r$-projection profile of $u$ on $A$ in $G$ can be uniquely decoded from the $r$-distance profile of $(u,0)$ on $B$ in $H$.

\begin{claim}\label{cl:encode}
For every $u\in V(G)\setminus A$ and $v\in A$, the value $\rho^G_r[u,A](v)$ is equal to the smallest index~$i$ such that $\pi^H_r[(u,0),B]((v,i))=i$, or $\infty$ if no such index exists.
\end{claim}
\begin{clproof}
Observe that, for each $i\in \{0,1,\ldots,r\}$, the distance from $(u,0)$ to the $i$-th layer of~$H$ is at least $i$, hence we always have $\pi^H_r[(u,0),B]((v,i))\geq i$.
To prove the claim, we show inequalities in two directions.

In one direction, suppose that $P$ is a shortest $A$-avoiding path connecting $u$ and $v$.
Let $P=(u=x_0,x_1,\ldots,x_i=v)$, where the length of $P$ is $i$, and suppose $i\leq r$.
Then it follows that $(x_0,0)(x_1,1)\ldots(x_i,i)$ is a path in $H$, certifying that $\pi^H_r[(u,0),B]((v,i))\leq i$.
Together with the inequality from the previous paragraph this implies that $\pi^H_r[(u,0),B]((v,i))=i$.

In the other direction, suppose $\pi^H_r[(u,0),B]((v,i))=i$ for some $i$, and let $Q$ be the path in~$H$ certifying this.
Since $Q$ has length $i$ and connects $(u,0)$ with a vertex of $i$-th layer, it follows that~$Q$ has to have the form $(x_0,0),(x_1,1)\ldots,(x_{i-1},i-1),(x_i,i)$,
where $u=x_0,x_1,\ldots,x_{i-1},x_i=v$ are vertices of $H$. 
By the construction of $H$, we infer that $\{x_{i-1},x_i\}$ is an edge in $G$ for all $i=1,2,\ldots,r$, and moreover $x_0,x_1,\ldots,x_{i-1}$ do not belong to $A$.
Therefore, $(x_0,x_1,\ldots,x_i)$ is a walk of length $i$ in $G$ connecting $u$ with $v$, where only the last vertex belongs to $A$.
It follows that in $G$ there is an $A$-avoiding path of length at most $i$ connecting $u$ with $v$, which certifies that $\rho^G_r[u,A](v)\leq i$.
\end{clproof}

By~\cref{cl:encode}, we have
\[\projprof_r(G,A)\leq \profnum_r(H,B).\]
On the other hand, by~\cref{thm:main-profiles} we have
\[\profnum_r(H,B)\leq \fnei'(r,\epsilon)\cdot |B|^{1+\epsilon}=\fnei'(r,\epsilon)\cdot (r+1)^{1+\epsilon}\cdot |A|^{1+\epsilon},\]
where function $\fnei'(r,\epsilon)$ is given for the (nowhere dense) class $\CCC^{r+1}$.
However, $\CCC^{r+1}$ depends only on $\CCC$ and $r$, hence $\fnei'(r,\epsilon)$ depends only on $\CCC$, $r$, and $\epsilon$.
Therefore, we can set $\fproj(r,\epsilon)=\fnei'(r,\epsilon)\cdot (r+1)^{1+\epsilon}$.
\end{proof}

We next recall the main tool for projections proved by Drange et al.~\cite{drange2016kernelization}, namely the {\em{Closure Lemma}}.
Intuitively, it says that any vertex subset $A\subseteq V(G)$ can be ``closed'' to a set $\cl_r(A)$ that is not much larger than $A$, such that all $r$-projections on $\cl_r(A)$ are small.
The next lemma follows from a straightforward adaptation of the proof of Drange et al.

\begin{lemma}[Lemma 2.9 of~\cite{drange2016kernelization}, adjusted]\label{lem:closure}
There is a function $\fcl(r,\epsilon)$ and a polynomial-time algorithm that, given $G\in \CCC$, $X\subseteq V(G)$, $r\in \N$ and $\epsilon>0$, computes the {\em{$r$-closure}} of $X$, denoted $\cl_r(X)$
with the following properties. 
\begin{itemize}
  \item $X\subseteq \cl_r(X)\subseteq V(G)$;
  \item $|\cl_r(X)|\leq \fcl(r,\epsilon)\cdot |X|^{1+\epsilon}$; and
  \item $|M_r^G(u,\cl_r(X))|\leq \fcl(r,\epsilon)\cdot |X|^{\epsilon}$ for each $u\in V(G)\setminus \cl_r(X)$.
\end{itemize}
\end{lemma}

Let us elaborate on the adaptation. Drange et al.\ work in the bounded expansion setting, and they give the following upper bounds on $|\cl_r(X)|$ and $|M_r^G(u,\cl_r(X))|$, 
for each $u\in V(G)\setminus X$:
\[|\cl_r(X)|\leq ((r-1)\xi+1)\cdot |X|\quad\textrm{and}\quad |M_r^G(u,\cl_r(X))|\leq \xi(1+(r-1)\xi),\]
where $\xi=\lceil 2\grad_{r-1}(\CCC)\rceil$; here $G$ is assumed to belong to a class of bounded expansion $\CCC$. In the nowhere dense setting, $\grad_{r-1}(\CCC)$ is not bounded by a constant,
but~\cref{lem:gradbound} ensures that the edge density of any $r$-shallow minor $H$ of a graph from $\CCC$ is bounded by $\fgrad(r,\epsilon)\cdot |V(H)|^\epsilon$. In the proof of Drange et al.,
the only place where the parameter $\xi$ comes from is to estimate the edge density in a graph that is an $(r-1)$-shallow minor of the original graph $G$, in which, moreover, the number of vertices is
polynomial in $|X|$. It follows by~\cref{lem:gradbound} that in the nowhere dense setting, we can substitute $\xi$ with $f(r,\epsilon)\cdot |X|^{\epsilon}$ 
for some function $f$, thus obtaining the statement of~\cref{lem:closure} after rescaling $\epsilon$.

Finally, we need another lemma from Drange et al., called the {\em{Short Paths Closure Lemma}}. 
Here, it is shown that a set can be closed without blowing up its size too much, so that short distances between its elements are preserved in the subgraph induced by the closure.
Again, the proof is a straightforward adaptation of the proof of Drange et al.

\begin{lemma}[Lemma 2.11  of~\cite{drange2016kernelization}, adjusted]\label{lem:pathsclosure}
There is a function $\fpaths(r,\epsilon)$ and a polynomial-time algorithm which on input $G\in \CCC$, $X\subseteq V(G)$, $r\in \N$, and $\epsilon>0$, 
computes a superset $X'\supseteq X$ of vertices with the following properties:
\begin{itemize}
\item whenever $\dist_G(u,v)\leq r$ for $u,v\in X$, then $\dist_{G[X']}(u,v)=\dist_G(u,v)$; and
\item $|X'|\leq \fpaths(r,\epsilon)\cdot |X|^{1+\epsilon}$.
\end{itemize}
\end{lemma}

The upper bound on $|X'|$ given by Drange et al. is $|X'|\leq Q_{r-1}(\grad_{r-1}(\CCC))\cdot |X|$, for some polynomial $Q_{r-1}$.
Again, if $\CCC$ is only assumed to be nowhere dense, $\grad_{r-1}(\CCC)$ is not bounded, but the analysis of the proof of Drange et al. reveals that the only place this factor comes
from is an application of the Closure Lemma at the beginning of the proof. After replacing this lemma with its counterpart for a nowhere dense class (\cref{lem:closure}),
factor $Q_{r-1}(\grad_{r-1}(\CCC))$ is substituted with a polynomial of $\fcl(r,\epsilon)\cdot |X|^{\epsilon}$, which is of the form $\fpaths(r,\epsilon)\cdot |X|^{\epsilon}$ after rescaling~$\epsilon$.

\paragraph{Towards the proof.} With all the tools gathered, we can now proceed to the proof of~\cref{thm:main-kernel}. For convenience, we recall its statement.

\setcounter{theorem}{1}

\begin{theorem}[restated]
  Let $\CCC$ be a fixed nowhere dense class of graphs, let $r$ be a fixed positive integer, and let $\epsilon>0$ be any fixed real.
  Then there is a polynomial-time algorithm that, given a graph $G\in \CCC$ and a positive integer $k$,
  returns a subgraph $G'\subseteq G$ and a vertex subset $Z\subseteq V(G')$ with the following properties:
  \begin{itemize}
  \item there is a set $D\subseteq V(G)$ of size at most $k$ which 
  $r$-dominates $G$ if and only if there is a set $D'\subseteq V(G')$
  of size at most~$k$
  which $r$-dominates $Z$ in $G'$; and
  \item $|V(G')|\leq \fker(r,\epsilon)\cdot k^{1+\epsilon}$, for some function $\fker(r,\epsilon)$ depending only on the class $\CCC$.
  \end{itemize}
\end{theorem}

\setcounter{theorem}{9}

For the rest of this section let us fix constants $r\in \N$ and $\epsilon>0$; they will be used implicitly in the proofs.
Our proof follows the general strategy proposed by Drange et al.~\cite{drange2016kernelization} for the bounded expansion setting, 
however we replace parts that break in the nowhere dense setting
with the new results. The crucial new tool we use is the near-linear bound on the neighborhood complexity in nowhere dense graphs, more precisely the projection variant (\cref{thm:main-profiles}).
However, there are also other places where a fix is necessary. Here, the polynomial bounds on uniform quasi-wideness (\cref{thm:uqw}) 
proved recently by Kreutzer et al.~\cite{siebertz2016polynomial} will appear useful.

Let us recall some terminology from Drange et al.~\cite{drange2016kernelization}, which is essentially also present in the approach of Dawar and Kreutzer~\cite{DawarK09}. 
For an integer $r\in \N$, graph $G$, and vertex subset $Z\subseteq V(G)$, we say that a subset of vertices $D$ is a {\em{$(Z,r)$-dominator}}
if $Z\subseteq N^G_r(D)$. We write $\mathbf{ds}_r(G,Z)$ for the smallest
$(Z,r)$-dominator in $G$ and $\mathbf{ds}_r(G)$ for the smallest $(V(G),r)$-dominator in $G$. 
The crux of the approach of Drange et al. is to perform kernelization in two phases: first compute a small {\em{$r$-domination core}}, and then reduce the size of the graph in one step.

\begin{definition}
An \emph{$r$-domination core} of $G$ is a subset $Z\subseteq V(G)$ such that every minimum size $(Z,r)$-dominator is also a distance-$r$ dominating set in $G$. 
\end{definition}

We shall prove the following analogue of Theorem 4.11 of~\cite{drange2016kernelization}.

\begin{lemma}\label{lem:findcore}
There exists a function $\fcore(r,\epsilon)$ and a polynomial-time algorithm that, given a graph $G\in \CCC$ and integer $k\in \N$, 
either correctly concludes that $G$ cannot be $r$-dominated by $k$ vertices, or finds an $r$-domination core $Z\subseteq V(G)$ of $G$ of size at most $\fcore(r,\epsilon)\cdot k^{1+\epsilon}$.
\end{lemma}

Starting with $Z=V(G)$, which is clearly an $r$-domination core of $G$, we try
to iteratively remove vertices from $Z$ while preserving the property
that $Z$ is an $r$-domination core of $G$.  More precisely, we show the
following lemma, which is the analogue of Theorem 4.12
of~\cite{drange2016kernelization} and of Lemma 11
of~\cite{DawarK09}. Observe that \cref{lem:findcore} follows by the
applying~\cref{lem:findredundantvertex} iteratively until the size of the core is reduced to at most $\fcore(r,\epsilon) \cdot k^{1+\epsilon}$.
This iteration is performed at most $n$ times leading to an additional
factor~$n$ in the running time in~\cref{lem:findcore}.

\begin{lemma}\label{lem:findredundantvertex}
There exists a function $\fcore(r,\epsilon)$ and a polynomial-time algorithm that, given a graph $G\in \CCC$, an integer $k\in \N$, and an $r$-domination core
$Z$ of $G$ with $|Z|>\fcore(r,\epsilon) \cdot k^{1+\epsilon}$, either correctly concludes that $Z$ cannot be $r$-dominated by $k$ vertices, or finds a vertex $z\in Z$ such that $Z\setminus\{z\}$ is still an $r$-domination
core of~$G$. 
\end{lemma}

As in Drange et al., the first step of the proof of~\cref{lem:findredundantvertex} is to find a suitable approximation of a $(Z,r)$-dominator.
For this, Drange et al.\ use a quite complicated scheme of iteratively applying the approximation algorithm of Dvo\v{r}\'ak~\cite{dvovrak2013constant}, 
because for the next step they need a dual object -- a large $2r$-independent set -- which is provided by the algorithm of Dvo\v{r}\'ak as a certificate of approximation.
The need of ensuring that this iterative construction finishes in a constant number of rounds is another piece that breaks in the nowhere dense setting.
We circumvent this problem by obtaining a suitably large
$2r$-independent set using the polynomial bounds on uniform
quasi-wideness (\cref{thm:uqw}) instead of the algorithm of Dvo\v{r}\'ak.
Therefore, we can rely on the classic $\Oof(\log \mathsf{OPT})$-approximation of Br\"onnimann and Goodrich~\cite{bronnimann1995almost} for the {\sc{Hitting Set}} problem in 
set families of bounded VC-dimension, as explained in the next lemma.

\begin{lemma}\label{lem:dvorak}
There is a function $\fapx(r,\epsilon)$ and a polynomial-time algorithm that, given $G\in \CCC$ and $Z\subseteq V(G)$, 
computes a $(Z,r)$-dominator of size at most $\fapx(r,\epsilon)\cdot k^{1+\epsilon}$, where $k$ is the size of a minimum $r$-dominating set of $Z$. 
\end{lemma}
\begin{proof}
Let $\FFF$ be a family of subsets of $V(G)$ defined as $\FFF=\{N^G_r[u]\colon u\in Z\}$.
Family $\FFF$ is a subfamily of the family of all closed $1$-neighborhoods in the graph~$G_{\leq r}$, as defined in~\cref{thm:adleradler}.
By \cref{thm:adleradler}, $G_{\leq r}$ has VC-dimension at most $c=c(r)$, where $c$ depends on $\CCC$ and $r$,
hence also the VC-dimension of $\FFF$ is bounded by $c$.

It is clear that $(Z,r)$-dominators in $G$ correspond to hitting sets of $\FFF$, that is, subsets of $V(G)$ that intersect every set of $\FFF$.
Br\"onnimann and Goodrich~\cite{bronnimann1995almost} gave a polynomial-time algorithm that, given a set family of VC-dimension at most $c$ and an integer $k$,
either correctly concludes that there is no hitting set of size at
most $k$, or finds a hitting set of size at most $\Oof(ck\log (ck))$.
More precisely, the running time of this pseudopolynomial algorithm
becomes polynomial, with the exponent not dependent on~$\CCC$, if we
precompute $\FFF$, store it in a binary matrix,
and implement subsystem and witness oracles, as defined in~\cite{bronnimann1995almost}, as linear scans over all the elements of $\FFF$.
Finally, we can choose the function $\fapx(r,\epsilon)$ so that the upper bound on the size of the obtained hitting set, being at most $\Oof(ck\log (ck))$, is upper bounded by $\fapx(r,\epsilon)\cdot k^{1+\epsilon}$.
Hence, the set provided by the algorithm of~\cite{bronnimann1995almost} applied to $\FFF$ can be returned.
\end{proof}

We are now ready to prove~\cref{lem:findredundantvertex}.

\begin{proof}[of~\cref{lem:findredundantvertex}]
  The function $\fcore(r,\epsilon)$ will be defined in the course of the proof.
  For convenience, for now we assume that $|Z|>\fcore(r,\epsilon) \cdot k^{1+C\epsilon}$ for some large constant $C$, for at the end we will rescale $\epsilon$ accordingly.

  We first apply~\cref{lem:dvorak}, to either conclude that there is
  no $(Z,r)$-dominator of size at most~$k$ or to compute a $(Z,r)$-dominator
  $X$ of size at most $\fapx(r,\epsilon)\cdot k^{1+\epsilon}$. 
  This application takes polynomial time.
  In the first case we can reject the instance, hence assume that we are in the second case.

  We apply the algorithm of~\cref{lem:closure} to the set $X$ and distance parameter $3r$,
  thus computing its closure $\cl_{3r}(X)$, henceforth denoted by $X_{\cl}$.
  By~\cref{lem:closure}, we have
  \begin{eqnarray*}
    |X_{\cl}| & \leq & \fcl(3r,\epsilon)\cdot |X|^{1+\epsilon}\\
             & \leq & \fcl(3r,\epsilon)\cdot \fapx(r,\epsilon)^{1+\epsilon}\cdot k^{1+2\epsilon+\epsilon^2}\leq \fcl(3r,\epsilon)\cdot \fapx(r,\epsilon)^{1+\epsilon}\cdot k^{1+3\epsilon},
  \end{eqnarray*}
  and, for every $u\in V(G)\setminus X_\cl$,
  \begin{eqnarray*}
  |M^G_{3r}(u,X_\cl)| & \leq & \fcl(3r,\epsilon)\cdot |X|^{\epsilon}\\
                      & \leq & \fcl(3r,\epsilon)\cdot \fapx(r,\epsilon)^{\epsilon}\cdot k^{2\epsilon+\epsilon^2}\leq \fcl(3r,\epsilon)\cdot \fapx(r,\epsilon)^{\epsilon}\cdot k^{3\epsilon}.
  \end{eqnarray*}
  
  We now classify the elements of $Z\setminus X_{\cl}$ according to their
  $3r$-projection profiles on $X_{\cl}$. More precisely, let us define an equivalence relation $\sim$ on $Z\setminus X_{\cl}$ as follows:
  \[u\sim v\quad\textrm{if and only if}\quad \rho^G_{3r}[u,X_\cl]=\rho^G_{3r}[v,X_\cl].\]
  By~\cref{lem:projection-complexity}, the number of equivalence classes of $\sim$ is bounded as follows: 
  \begin{eqnarray*}
  \indx(\sim) & \leq & \fnei(3r,\epsilon)\cdot |X_\cl|^{1+\epsilon}\\
              & \leq & \fnei(3r,\epsilon)\cdot \fcl(3r,\epsilon)^{1+\epsilon}\cdot \fapx(r,\epsilon)^{(1+\epsilon)^2}\cdot k^{1+7\epsilon}.
  \end{eqnarray*}
  Note that the partition of $V(G)$ into the equivalence classes of $\sim$ can be computed in polynomial time, by just computing the $3r$-projection profile for each vertex using breadth-first search,
  and then comparing the profiles pairwise.
  
  Let $p(r)$ and $s(r)$ be the functions provided by~\cref{thm:uqw} for the class $\CCC$.
  Denote 
  \begin{eqnarray*}
  \alpha & = & \fcl(3r,\epsilon)\cdot\fapx(r,\epsilon)^{\epsilon}\cdot k^{3\epsilon}+s(2r)+1\quad\textrm{and}\\
  \beta & = & \alpha \cdot (r+2)^{s(r)}+1.
  \end{eqnarray*}
  From the above inequalities on $|X_\cl|$ and $\indx(\sim)$, it follows that setting $C=7+3\cdot p(2r)$, we can fix the function $\fcore(r,\epsilon)$ so that the following inequality is always satisfied: 
  \[\fcore(r,\epsilon)\cdot k^{1+C\epsilon} \geq |X_\cl|+\indx(\sim)\cdot \beta^{p(2r)}.\]
  Since we assumed that $|Z|>\fcore(r,\epsilon)\cdot k^{1+C\epsilon}$, it follows that 
  $|Z\setminus X_\cl|>\indx(\sim)\cdot \beta^{p(2r)}$.
  Hence, by the pigeonhole principle, there exists an equivalence class $\kappa$ of $\sim$ that contains more than $\beta^{p(2r)}$ vertices.
  By applying the algorithm of~\cref{thm:uqw} to any subset of $\kappa$ of size exactly $\beta^{p(2r)}$,
  we find sets $S\subseteq V(G)$ and $L\subseteq \kappa\setminus S$ such that $|S|\leq s(r)$, $|L|\geq \beta$, and $L$ is $2r$-independent in $G-S$.
  This application takes time $\Oof(r\cdot c\cdot \beta^{p(2r)\cdot (c+6)}\cdot |V(G)|^2)$.
  Provided $\epsilon$ satisfies $3\cdot p(2r)\cdot (c+6)\cdot \epsilon<1$, which we can assume without loss of generality, we have that $\beta^{p(2r)\cdot (c+6)}\leq \Oof(k)$;
  here, the constants hidden in the $\Oof(\cdot)$-notation may depend on $\CCC$.
  Since we can assume that $k\leq |V(G)|$, this application of the algorithm of~\cref{thm:uqw} takes then time $\Oof(|V(G)|^3)$, which is polynomial with the degree independent of $\CCC$.

  We now classify the elements of $L$ according to their $r$-distance profiles on $S$.
  Note that $|l|\leq \beta$ and that the number of $r$-distance profiles on $S$ is bounded by $(r+2)^{|S|}\leq (r+2)^{s(r)}$.
  Since $\beta>\alpha \cdot (r+2)^{s(r)}$, by the pigeonhole principle we infer that there is a subset $R\subseteq L$ of size $|R|>\alpha$ such that
  \[\pi_r^G(v,S)=\pi_r^G(w,S)\quad\textrm{for all}\quad v,w\in R.\]
  
  At this point the situation is almost exactly as in Lemma 3.8 of~\cite{drange2016kernelization}. More precisely,
  every vertex of $R$ is an irrelevant dominatee, i.e., it can be excluded from $Z$ without spoiling the property that $Z$ is a domination core.
  For the sake of completeness, we repeat the argumentation from~\cite{drange2016kernelization}, which is based on an exchange argument. 
  Let $z$ be any vertex of $R$, which also belongs to $Z$ due to $R\subseteq L\subseteq \kappa\subseteq Z\setminus X_\cl$,
  and let $Z'=Z\setminus \{z\}$.
  
  \begin{claim}\label{cl:exchange}
  The set $Z'$ is also a domination core of $G$.
  \end{claim}
  \begin{clproof} 
  Suppose $D$ is a minimum-size $(Z',r)$-dominator in $G$.
  Consider first the case when $z$ is $r$-dominated by $D$.
  Then $D$ is in fact a $(Z,r)$-dominator, and it must be a minimum-size $(Z,r)$-dominator since $Z'\subseteq Z$ and $D$ is a minimum-size $(Z',r)$-dominator.
  Since $Z$ is an $r$-domination core by assumption, we infer that $D$ is an $r$-dominating set of $G$, hence we are done in this case.
  Thus, for the rest of the proof we assume that $D$ does not $r$-dominate $z$, and we aim at a contradiction.
  
  Consider any vertex $x\in R\setminus \{z\}$. Since $x\in Z'$, we have that $x$ is $r$-dominated by $D$.
  For each such $x$, let us select any vertex $d(x)\in D$ that $r$-dominates $x$.
  Further, let $P(x)$ be an arbitrary path of length at most $r$ that connects $d(x)$ with $x$.
  
  We first prove that for each $x\in  R\setminus \{z\}$, the path $P(x)$ does not traverse any vertex of $X_\cl\cup S$.
  For suppose otherwise, and assume that $y$ is the first (closest to $x$) vertex of $P(x)$ that belongs to $X_\cl$.
  Suppose further that the length of the prefix of $P(x)$ from $x$ to $y$ is $\ell$, for some $0\leq \ell \leq r$;
  thus $\dist_G(x,y)\leq \ell$ and $\dist_G(y,d(x))\leq r-\ell$. If we had $y\in S$, then since all vertices of $R$ have the same distance profile on $S$,
  we would have $\dist_G(z,y)=\dist_G(x,y)\leq \ell$, implying
  \[\dist_G(z,d(x))\leq \dist_G(z,y)+\dist_G(y,d(x))\leq \ell+r-\ell=r,\]
  so $z$ would be $r$-dominated by $d(x)\in D$. On the other hand, if we had $y\in X_\cl$, then the prefix of $P(x)$ from $x$ to $y$, being an $X_\cl$-avoiding path, certifies that $\rho^G_{3r}[x,X_\cl](y)\leq \ell$.
  Vertices of $R\subseteq \kappa$ have the same projection profiles on $X_\cl$, hence also $\rho^G_{3r}[z,X_\cl](y)=\rho^G_{3r}[x,X_\cl](y)\leq \ell$, implying in particular that $\dist_G(z,y)\leq \ell$.
  Then the same reasoning as in the previous case leads to the contradiction with the supposition that $z$ is not $r$-dominated by $D$.
  
  We conclude that the paths $\{P(x)\colon x\in R\setminus \{z\}\}$ are entirely contained in $G-(X_\cl\cup S)$. 
  Since $R\setminus \{z\}$ is $2r$-independent in $G-S$, and each path $P(x)$ has length at most $r$, we conclude that
  vertices $d(x)$ for $x\in R\setminus \{z\}$ are pairwise different.
  In particular, if we denote $D_{\mathrm{sell}}=\{d(x)\colon x\in R\setminus \{z\}\}$, then
  \[|D_{\mathrm{sell}}|=|R\setminus \{z\}|>\alpha-1 = \fcl(3r,\epsilon)\cdot\fapx(r,\epsilon)^{\epsilon}\cdot k^{3\epsilon}+s(2r).\]
  Denote $D_{\mathrm{buy}}=M_{3r}(z,X_\cl)\cup S$. Since $|S|\leq s(2r)$ and the $3r$-projection of each vertex of $V(G)\setminus X_\cl$ onto $X_\cl$ has size at most 
  $\fcl(3r,\epsilon)\cdot\fapx(r,\epsilon)^{\epsilon}\cdot k^{3\epsilon}$, we infer that
  \[|D_{\mathrm{buy}}|\leq \fcl(3r,\epsilon)\cdot\fapx(r,\epsilon)^{\epsilon}\cdot k^{3\epsilon}+s(2r).\]
  Thus $|D_{\mathrm{sell}}|>|D_{\mathrm{buy}}|$.
  Denote $D'=(D\setminus D_{\mathrm{sell}})\cup D_{\mathrm{buy}}$; then we have $|D'|<|D|$.
  We claim that $D'$ is still a $(Z',r)$-dominator in $G$, which will contradict the assumption that $D$ is a minimum-size $(Z',r)$-dominator in $G$.
  
  For the sake of contradiction, suppose there is some vertex $a\in Z'$ that is not $r$-dominated by $D'$.
  Since $D$ is a $(Z',r)$-dominator and $D\setminus D'\subseteq D_{\mathrm{sell}}$, it follows that there is some vertex $x\in R\setminus \{z\}$ 
  such that $d(x)$ was the vertex of $D$ that $r$-dominated $a$, i.e., $\dist_G(d(x),a)\leq r$.
  Let $Q_1$ be any path of length at most $r$ connecting $d(x)$ and $a$.
  Futhermore, since $a\in Z$ and $X_\cl$ is a $(Z,r)$-dominator, there is some path $Q_2$ of length at most $r$ that connects $a$ with a vertex of $X_\cl$.
  
  Consider now the concatenation of paths $P(x)$, $Q_1$, and $Q_2$. This is a walk of length at most $3r$ connecting $x$ with a vertex of $X_\cl$.
  Let $b$ be the first vertex of this walk (from the side of~$x$) that belongs to $X_\cl\cup S$.
  We have that $b$ does not lie on $P(x)$, since we argued that $P(x)$ is disjoint with $X_\cl\cup S$.
  Therefore, $b$ lies on $Q_1$ or $Q_2$, however every vertex lying on any of these paths is at distance at most $r$ from $a$.
  Hence $\dist_G(a,b)\leq r$, so it remains to argue that $b\in D'$.
  This is obvious if $b\in S$ since $S$ was explicitly included in $D_{\mathrm{buy}}$, so assume that $b\in X_\cl$.
  Then the prefix of the considered walk from $x$ to $b$ avoids $X_\cl$, which certifies that $b\in M_{3r}(x,X_\cl)$.
  However, $x,z\in \kappa$ and all vertices of $\kappa$ have the same $3r$-projection profile on $X_\cl$, so in particular the same $3r$-projection on $X_\cl$.
  We conclude that $b\in M_{3r}(z,X_\cl)\subseteq D_{\mathrm{buy}}$, which finishes the proof.
\end{clproof}

\cref{cl:exchange} ensures us that vertex $z$ is an irrelevant dominatee that can be returned by the algorithm.
Note that we were able to find $z$ provided $|Z|>\fcore(r,\epsilon)\cdot k^{1+C\epsilon}$ for some constant~$C$ depending on $\CCC$ and $r$.
Hence, we conclude the proof by rescaling $\epsilon$ to $\epsilon/C$ throughout the reasoning.
\end{proof}

Finally, having computed a suitably small domination core, we can construct a kernel from it.
This part of reasoning is exactly the same as in~\cite{drange2016kernelization}.

\begin{lemma}\label{lem:core-kernel}
There exists a function $\ffin(r,\epsilon)$ and a polynomial-time algorithm that, given a graph $G\in \CCC$ and an $r$-domination core $Z\subseteq V(G)$ of $G$,
computes a graph $G'$ with at most $\ffin(r,\epsilon)\cdot |Z|^{1+\epsilon}$ vertices such that $Z\subseteq V(G')$ and $\ds_r(G',Z)=\ds_r(G,Z)$.
\end{lemma}
\begin{proof}
  First, we classify the vertices of $V(G)\setminus Z$ according to their $r$-projections onto $Z$.
  Precisely, consider an equivalence relation $\sim$ on $V(G)\setminus Z$ defined as follows:
  \[u\sim v\quad\textrm{if and only if}\quad \rho^G_{r}[u,Z]=\rho^G_{r}[v,Z].\]
  By~\cref{lem:projection-complexity}, the equivalence relation $\sim$ has at most $\fproj(r,\epsilon)\cdot |Z|^{1+\epsilon}$ equivalence classes.
  For each equivalence class $\kappa$ of $\sim$, let us arbitrarily select a vertex $v_\kappa$ belonging to $\kappa$.
  Let us define:
  \[A=Z\cup \{v_\kappa\ \colon\ \kappa\textrm{ is an equivalence class of }\sim\}.\]
  Thus we have
  \[|A|\leq |Z|+\indx(\sim)\leq (1+\fproj(r,\epsilon))\cdot |Z|^{1+\epsilon}.\]
  Now, we apply~\cref{lem:pathsclosure} to $G$, the set $A$, distance parameter $r$, and $\epsilon$.
  Thus we compute a superset $A'\supseteq A$ with the following conditions satisfied:
  for every pair of vertices $u,v\in A$, if $\dist_G(u,v)\leq r$ then $\dist_G(u,v)=\dist_{G[A']}(u,v)$.
  Moreover, we have 
  \[|A'|\leq \fpaths(r,\epsilon)\cdot |A|^{1+\epsilon}\leq \fpaths(r,\epsilon)\cdot (1+\fproj(r,\epsilon))^{1+\epsilon}\cdot |Z|^{1+3\epsilon}.\]
  Now, we claim that the algorithm can output the graph $G[A']$; the correctness of this step follows from the following.
  \begin{claim}\label{cl:correctness-closure}
  It holds that $\ds_r(G)=\ds_r(G[A'],Z)$.
  \end{claim}
  The proof of~\cref{cl:correctness-closure} is the same as the proof of Claims 3.12 and 3.14 in~\cite{drange2016kernelization}, and hence we omit it and refer the reader to~\cite{drange2016kernelization} instead.
  Let us only remark that the crucial observation is that if $D$ is a minimum-size dominating set in $G$, then  
  \[D'=(D\cap Z)\cup \{v_\kappa\ \colon\ \kappa\textrm{ is an equivalence class of }\sim\textrm{ such that }\kappa\cap D\neq \emptyset\}.\]
  is also a minimum-size dominating set in $G$, while $D'\subseteq A'$.
  We conclude by rescaling $\epsilon$ by factor~$3$ in order to achieve an upper bound $|A'|\leq \ffin(r,\epsilon)\cdot |Z|^{1+\epsilon}$ for some function $\ffin(r,\epsilon)$. 
\end{proof}

\cref{thm:main-kernel} follows by combining~\cref{lem:findcore} and~\cref{lem:core-kernel}, and rescaling $\epsilon$ by factor $3$.

\section{Conclusion}\label{sec:conclusion}

In this paper we proved that for monotone graph classes, nowhere denseness is equivalent to admitting an almost linear upper bound on the neighborhood complexity of any subset of vertices of a graph from the class.
As a concrete algorithmic application of this characterization, we gave an almost linear kernel for the {\sc{Distance-$r$ Dominating Set}} problem on any nowhere dense class of graphs and any $r\in \N$.
This yields a dichotomy theorem for the parameterized complexity of {\sc{Distance-$r$ Dominating Set}} on monotone classes: there is a sharp jump in the complexity exactly on 
the boundary between nowhere and somewhere denseness.

Our study contributes to the general research programme of understanding the connections between the combinatorial, algorithmic, and logical aspects of the notions of abstract sparsity;
this programme is currently under intensive development. In particular, we believe that this work together with~\cite{siebertz2016polynomial} provides fundamental applications of model-theoretic tools borrowed
from the stability theory to the combinatorial and algorithmic theory of nowhere dense classes. Conceptually, we think that the most important technical insight can be summarized as follows.
When working with a graph $G$ from a fixed nowhere dense class $\CCC$, one often would like to define auxiliary objects that are not necessarily sparse; examples include
the hypergraph of balls of radius $r$ around every vertex, or the hypergraph of weak $r$-reachability sets that we encountered in~\cref{sec:neighborhood-complexity}. 
While we cannot directly apply the sparsity toolbox to these objects due to their dense nature,
we can still infer some of their properties, predominantly the boundedness of the VC-dimension, from the stability of the graph class $\CCC$. This, together with the Sauer-Shelah Lemma, is often enough
to get a refined combinatorial understanding that leads to conclusions for the original sparse graph $G$. More concretely, as~\cite{siebertz2016polynomial} and this work shows, 
a bound on the VC-dimension combined with the Sauer-Shelah Lemma is often the right tool for turning exponential blow-ups in combinatorial estimations into polynomial, which enables lifting
reasonings from the bounded expansion setting to the nowhere dense setting.

\bibliographystyle{abbrv}

\end{document}